\documentclass[11pt,a4paper]{article}
\usepackage[top=3cm, bottom=3cm, left=3cm, right=3cm]{geometry}
\usepackage{lmodern,ulem}
\usepackage[utf8]{inputenc}
\usepackage[english]{babel}
\usepackage[T1]{fontenc}
\usepackage{dsfont}
\usepackage{amsmath}
\usepackage{todonotes}
\usepackage{amsfonts}
\usepackage{amssymb}
\usepackage{amsthm}
\usepackage{graphicx}
\usepackage{eurosym}
\usepackage[hidelinks]{hyperref}
\usepackage{cancel}
\usepackage{comment}
\usepackage{tikz}
\usepackage{float}
\usetikzlibrary{arrows}
\newcommand{\s}{\sigma}
\renewcommand{\b}{\beta}

\renewcommand{\a}{\alpha}
\newcommand{\e}{\varepsilon}
\newcommand{\A}{\mathcal{A}}
\newcommand{\E}{\mathbb{E}}

\newcommand{\R}{\mathds{R}}

\renewcommand{\r}{\rho}
\newcommand{\g}{\gamma}
\renewcommand{\o}{\omega}

\newtheorem{theorem}{Theorem}[section]
\newtheorem{proposition}{Proposition}[section]
\newtheorem{remark}{Remark}[section]
\newtheorem{lemma}{Lemma}[section]
\newtheorem{corollary}{Corollary}[section]
\newtheorem{assumption}{Assumption}[section]
\theoremstyle{definition}
\newtheorem{definition}{Definition}[section]
\DeclareMathOperator{\essinf}{ess\,inf}

\title{Price impact and long-term profitability of energy storage\footnote{We thank Olivier Féron (Electricité de France) and Sébastien Soleille (BNP Paribas) for insightful comments on an earlier draft of this paper. This study was carried out in the framework of ``Energy for Climate'' interdisciplinary research center and was supported financially by the ``Decarbonize energy'' program of the Institut Polytechnique de Paris, as well as by the FIME Research Initiative of the Europlace Institute of Finance.}}
\author{Roxana Dumitrescu, Redouane Silvente and Peter Tankov\\CREST, ENSAE, Institut Polytechnique de Paris}
\date{}

\numberwithin{equation}{section}
\begin{document}

\maketitle


\begin{abstract}
We study the price impact of storage facilities in electricity markets and analyze the long-term profitability of these facilities in prospective scenarios of energy transition. To this end, we begin by characterizing the optimal operating strategy for a stylized storage system, assuming an arbitrary exogenous price process. Following this, we determine the equilibrium price in a market comprising storage systems (acting as price takers), renewable energy producers, and conventional producers with a defined supply function, all driven by an exogenous demand process.
The price process is characterized as a solution to a fully coupled system of forward-backward stochastic differential equations, for which we establish existence and uniqueness under appropriate assumptions. 
We finally illustrate the impact of storage on intraday electricity prices through numerical examples and show how the revenues of storage agents may evolve in prospective energy transition scenarios from RTE, the French energy electricity network operator, taking into account both the increasing penetration of renewable energies and the self-cannibalization effect of growing storage capacity. We find that both the average revenues and the interquantile ranges increase in all scenarios, highlighting higher expected profits and higher risk for storage assets. 
\end{abstract}




\section{Introduction}
The transition to renewable energy production is a major challenge of the 21st century. According to the latest IPCC report \cite{IPCC}, even if no new fossil-fired power plants are constructed from now on, and the existing ones keep running until the end of their lifetime, the target of limiting global warming below $1.5^\circ$C will not be respected. Thus, conditions must be created for the massive replacement of fossil-fired power plants by renewable technologies. This is exacerbated by growing electricity demand:  although it fell slightly in 2020 owing to the COVID-19 pandemic, global electricity demand grew by 5\% in 2021 and by 2.5\%  in 2022, similar to the average annual growth rate of 2.6\% in the previous decade (2010-2021)\footnote{See Global Electricity Review 2021, 2022 and 2023, \href{https://ember-climate.org/}{ember-climate.org}.}. \textcolor{black}{In France, the French TSO RTE expects energy consumption to increase by $36\%$ by 2050, according to its Reference Scenario.}\footnote{https://rte-futursenergetiques2050.com/trajectoires/trajectoire-de-reference}



Given the intermittency of renewable energies, electricity storage devices play a crucial role in the energy transition. Multiple electricity storage technologies exist  (Pumped Hydroelectric Electricity Storage (PHES), battery storage with various battery systems, compressed air storage, hydrogen storage, supercapacitors and mechanical storage devices such as flywheels) and  can provide various services to the energy systems and users including energy arbitrage, frequency regulation, seasonal storage, peaker replacement, congestion management etc. We refer the reader to  \cite{schmidt2019projecting} for details of various storage technologies and services they can provide to energy systems, as well as their projected costs.

    The development of electricity storage in Western countries started in the second half of the last century with the build-up of PHES \cite{barbour_review_2016}. As of the present day, PHES account for approximately 86\%  of the global electricity storage infrastructure\footnote{See Renewables 2023 Global Status Report, \href{https://www.ren21.net/gsr-2023/}{ren21.net/gsr-2023/}}, with the remaining portion primarily consisting of battery-based systems. 
    \textcolor{black}{Modeling storage systems involves significant technical difficulties. From an engineering perspective, one must consider the discharge duration of the technology and the losses incurred during each charging cycle. For example, the discharge duration of pumped storage systems is relatively high (from several hours to several days), whereas Li-ion batteries allow for lower discharge times (less than an hour) \cite{elio2021review}. Additionally, some storage technologies  are affected by external factors, such as weather conditions for PHES (like rain or drought) and, to a lesser extent, temperature for batteries.}

Energy price arbitrage is the main source of revenues that storage agents can use to offset their investment costs. Nevertheless, the proliferation of storage units in the market may reduce price differences and available arbitrage opportunities, leading to lower than expected revenues and endagnering the profitability of units and, ultimately, the energy transition objectives through the so-called self-cannibalization effect. At the same time, increased penetration of intermittent renewable energy production may lead to higher volatility, higher price differences and, consequently, increased profits for energy storage. Therefore, the goal of this paper is to understand the impact of electricity storage devices on market prices of electricity in the presence of renewable generation. After studying the optimal behavior of a single price-taker storage agent in the intraday market, we analyze the impact of this optimal behavior on the price formation and derive the price impact of energy storage. 



We first build a model of a \textcolor{black}{price taker storage system} and determine its optimal strategy in the short term given an exogenous stochastic price process. We use the framework of stochastic optimal control to take into account weather-related random events, which affect the revenues of the storage agent, and which are modeled using Brownian and Poisson noise. This model can also consider a potential external energy source such as a renewable energy asset coupled with a storage device. We use a linear quadratic formulation for the problem solved by the storage agent to obtain a tractable equilibrium price model. This formulation leads to an explicit strategy in a closed-loop form that depends on the current electricity price and on the conditional expectation of future prices.


Next, we build a stylized equilibrium model of the electricity market comprising storage systems (acting as price takers), renewable energy producers, and conventional producers with a defined supply function, all driven by an exogenous demand process. In a toy model with no random sources in the market, we find an explicit solution for the electricity price at all times. In the general case the price process is characterized as the solution of a fully coupled system of forward-backward stochastic differential equations, for which we establish existence and uniqueness  under appropriate assumptions.


{\color{black}Finally, we study both short-term impact of storage on the market prices of electricity and its long-term impact on agents' revenues through numerical examples, calibrated to French intraday electricity market data and the prospective scenarios for the French electricity sector. 
In the short term, we find that an increase in storage capacity within the market leads to smaller differences between peak and off-peak prices and lower overall price volatility.  In the long term, we find that, ceteris paribus, increasing penetration of storage leads to a reduction of revenues of individual storage agents, because storage reduces price differences and, consequently, arbitrage opportunities. However, increasing penetration of intermittent renewable energies leads to higher price volatility, higher price differences and potentially higher revenues for storage agents. To understand whether the gains for storage agents related to increased renewable penetration will compensate the losses due to their price impact in the long term, we consider the reference scenarios for the evolution of the French electricity sector published by RTE, the French electricity network operator\footnote{See \href{https://rte-futursenergetiques2050.com/}{rte-futursenergetiques2050.com}}. Taking the renewable penetration projections and storage development projections from these scenarios, we find, under natural assumptions, that the expected revenues of a typical storage agent may grow by $150\%-400\%$ between 2020 and 2050, depending on the scenario, due to increased renewable penetration, even accounting for the losses due to self-cannibalization effect. However, the interquantile ranges of revenues increase in all scenarios, highlighting higher risk for storage agents. This higher risk is also apparent in the growth of price volatility, which increases, on average, by 120\% to 230\% depending on the scenario, despite the stabilizing effect of storage. 

 The paper is structured as follows. In the rest of this section we review the relevant literature and present the common notation for the rest of the paper. Section \ref{control.sec} presents the optimal control problem of a storage agent. The existence and uniqueness of the equilibrium price process are discussed in Section \ref{price.sec}. In Section \ref{shortterm.sec}
 we illustrate the short-term impact of storage on market prices and in Section \ref{longterm.sec} the long-term impact of storage on agents' revenues. Finally, in Appendix, we collect several proofs of technical results.}

\paragraph{Literature review}


Electricity storage has been widely studied with either a specific focus on the storage itself (see \cite{cruise2014optimal,cruise2018optimal,khalilisenobari2022optimal} or \cite{zhao_optimal_2009,steffen2016optimal} for an application to PHES), or with storage coupled with a production asset such as wind power plant \cite{collet2018optimal,korpaas2003operation,sousa2014impact} or oil power plant \cite{achdou2022class}. While in most papers the storage device is considered to be a price taker, with no direct impact on electricity prices,  some authors assume that the amount of storage is sufficient to actively impact electricity prices and consider the storage agent as a price maker \cite{ding2017optimal,sousa2014impact}. 

Typically, the primary objective is to determine the optimal operation strategy for a storage facility accounting for electricity price dynamics. This problem has been tackled with Lagrangian and non-convex optimisation techniques \cite{cruise2018optimal,khalilisenobari2022optimal}, bi-level linear programming \cite{cui2017bilevel,khalilisenobari2022optimal},  or  stochastic control (e.g., in \cite{kordonis2023optimal} using a value-iteration-type algorithm or in \cite{carmona_valuation_2010, shardin2017partially} with HJB equations). 
For an up-to-date overview of optimal control approaches to storage management, we refer the reader to \cite{machlev_review_2020}. The stochastic control approach makes it possible to account for random factors impacting storage devices, such as a random electricity price or random losses linked to external phenomena. However, with this approach it becomes more complex to take into account technical features of the storage system and to find a tractable solution. In this paper, we use linear quadratic stochastic control, which allows us both to take into account random factors in a rigorous manner and to derive an explicit solution for the optimal strategy of a storage system, so as to determine the equilibrium price dynamics later on.

Price formation in the electricity market has been studied by many authors over the past years, either with ad hoc market impact models or with equilibrium models. In the latter approach,  from a mathematical standpoint,  the price is a direct outcome of the equilibrium between producers competing to satisfy energy demand by offering their supply. Price formation in the intraday market, which will be the focus of this paper, has been studied in \cite{olivier2022price,aid2022equilibrium, aid2016optimal}. General dynamic stochastic equilibrium models often give rise to systems of Forward-Backward Stochastic Differential Equations (FBSDE) (see \cite{ma1999forward} for a general introduction, \cite{aid2022equilibrium} for an application to intraday market, \cite{shrivats2022mean} for an applcation to the certificate market with a McKean-Vlasov FBSDE, \cite{fujii2022mean,fujii2022equilibrium} for applications to equilibrium modeling in a more general setting), which will be the key tool in the proof of our main theorem. Another approach to price formation models is based on mean-field games. For instance \cite{aid_entry_2021} investigates an MFG framework involving both renewable and conventional producers (also refer to \cite{fujii2022equilibrium} and \cite{fujii2022mean} for a more general case not specifically applied to the electricity market).

The impact of energy storage on electricity prices has been explored  in various scenarios, including at the end-consumer level \cite{ahlert2010assessing} and at the wholesale level \cite{nyamdash2013impact} and it  has been recognized that energy storage has an impact on smoothing electricity prices \cite{gast2013impact} and can reduce price volatility \cite{hu2022impact,yang2016reducing}. Nevertheless these papers rarely address price formation in detail and mostly use simulation methods, engineering approaches or deterministic settings \cite{awad2014impact}. By contrast, our framework allows to prove the existence and uniqueness of the price process resulting from the equilibrium between producers in the market after computing the optimal strategy of a storage agent. 

{

\paragraph{Notations}

We fix a time horizon $T>0$. Define $\R^\star :=  \R \backslash \{0\} $, which we equip with a $\sigma$-finite measure $\nu$.
 Let $(\Omega, \mathcal{F}, \mathbb{P})$ be a probability space equipped with a right-continuous filtration $\mathbb{F}:=\{ 
 \mathcal{F}_t,\,\, t \in [0,T]\}$. Let $W$ be a one-dimensional $\mathbb{F}$-Brownian motion and let $N(dt,de)$ be an $\mathbb{F}$-Poisson random measure with compensator $dt \otimes \nu(de)$, supposed to be independent from $W$. We denote by $\Tilde{N}(dt,de)$ the compensated measure, i.e. $\Tilde{N}(dt,de):=N(dt,de)-dt\nu(de)$.
 
 We denote by $\mathcal{P}$ the predictable $\sigma$-algebra on $[0,T]\times \Omega$. We also use the following notation:
\begin{itemize}
\item[$\bullet$] $\mathcal{H}_d^2$ is the set of $\mathbb{R}^d$-valued predictable processes $\varphi$ such that  $\E[\int_0^T|\varphi_s|^2ds]<\infty$ (in the case when $d=1$, we will omit the subscript);
\item[$\bullet$] $\mathcal{H}^2_{\nu,d}$ is the set of $\mathbb{R}^d$-valued processes $l:(\omega,t,e) \in (\Omega \times [0,T] \times \textcolor{black}{\R^{\star}}) \mapsto l_t(\omega,e)$ which are predictable, i.e. $(\mathcal{P} \otimes \mathcal{B}(\mathbb{R}^\star),\mathcal{B}(\mathbb{R}))$-measurable, and such that $$\E \int_0^T \int_{\R^{\star}} |l(.,t,e)|^2 \nu(de)dt<\infty $$ (in the case when $d=1$, we will omit the subscript) 
\item $\mathcal{S}^2$ is the set of $\mathbb{R}$-valued adapted and RCLL processes $\phi$ such that    $\E[\sup_{0\leq t\leq T} |\phi_t|^2]<\infty$.
\item Let $\mathcal{M}^2$ be the set of square integrable martingales $M=(M_t)_{t \in [0,T]}$ with $M_0=0$. This is a Hilbert space equipped with the scalar product $(M,M')_{\mathcal{M}^2}:=\mathbb{E}[M_T M'_T]$, for $M,M' \in \mathcal{M}^2$. For each $M \in \mathcal{M}^2$, we set $\|M\|^2_{\mathcal{M}^2}:=\mathbb{E}[M_T^2]$.

\item Let $\mathcal{M}^{2,\perp}$ be the subspace of martingales of $H\in \mathcal{M}^2$ satisfying $\langle H,W \rangle_\cdot =0$ and such that for all predictable processes $l \in \mathcal{H}^2_\nu$, it holds that
$$\langle H,\int_0^\cdot \int_{\mathbb{R}^\star} l_s(e) \Tilde{N}(ds,de) \rangle_t=0, \,\, 0 \leq t \leq T,\,\, \text{a.s.}$$
\item $L^2(\mathcal{G})$ is the set of $\mathbb{R}$-valued $\mathcal{G}$-measurable random variables, with $\mathcal{G} \subset \mathcal{F}$.
\end{itemize}

\section{Optimal control of a storage} \label{control.sec}
 We consider an \textcolor{black}{electricity storage system} whose state process, denoted by $(Q^q_t)_t$, admits the following controlled dynamics:
    \begin{equation} \label{state}
        dQ^q_t=(-q_t+\kappa_t)dt+ \r_t dW_t+ \int_{\mathbb{R}^{\star}} \Xi(t,e)\widetilde N(dt,de), \,\, Q^q_0=\overline{Q}_0.
    \end{equation}
with $\overline{Q}_0>0$ a reference storage level. Here, $(q_t)$ plays the role of the control process and represents the withdrawal rate, while the real-valued processes $(\rho_t)_{0\leq t\leq T} \in \mathcal{H}^2$ and $\Xi \in \mathcal{H}_{\nu}^2$ describe the random amount of energy lost or gained because of external factors, either with small variations (Brownian noise) or through jumps (compensated Poisson random measure). For example, these factors can represent meteorological events that impact the tank level of a PHES, such as more or less intense rain, drought, etc. We also introduce a stochastic process $\kappa_t \in  \mathcal{H}^2$, which corresponds to a potential uncontrolled external source of energy supplying the storage, such as a wind power plant.

We assume that \textcolor{black}{the storage facility} operates exclusively in the intraday electricity market and that the main goal of the owner is to benefit from price arbitrages. In European intraday markets, all delivery periods are traded from the opening time of the intraday market to shortly before delivery. However, as shown in \cite{olivier2022price},  liquidity is only available during the last trading hours of each product. Therefore, we assume that the owner of the storage facility trades only shortly before delivery, at the last intraday price. We denote this price process by $(P_t)_{0\leq t\leq T}$ and refer to it as the real-time price.
 Throughout this study, each individual storage agent is supposed to be small enough to be considered as a price taker.

We formulate the optimization problem of the storage agent as a linear-quadratic optimal control problem with finite time horizon, which reads as follows:
\begin{equation} \label{Optimal}
     \underset{q \in \mathcal{A}}{\inf} \mathbb{E}\left[ \int_0^{T} \left\{-P_sq_s + \frac{\alpha}{2} q_s^2+\frac{\beta}{2}\left(Q^q_s-\overline{Q}_0\right)^2 \right\}ds + \frac{\gamma}{2}\left(Q^q_T-\overline{Q}_0\right)^2\right].
\end{equation}
Here, the first component, $-P_s q_s$, corresponds to the financial gain/loss of the storage agent from market transactions, while the second component, $\alpha q_t^2$ , represents a loss that occurs during each energy transfer, in other words, the energy discharge cost. While in reality the loss depends on the prevailing electricity price, this simplified model enables us to obtain an explicit expression for the solution of this optimal control problem.  

In our model, the state $(Q_t^q)$ is not explicitly constrained. Therefore, the purpose of the third term, $\frac{\beta}{2}(Q^q_s-\overline{Q}_0)^2$, is to financially penalize the agent for deviating too far from the initial reference level, in other words, this term represents the ``rental'' cost for a given capacity. Additionally, in the last term we penalize the final value of the storage level to ensure that the agent does not try to benefit from simply selling the energy just before the terminal time. 


For a single battery or a PHES facility, the quadratic objective functional does not a priori seem realistic since the capacity of these assets is bounded. However, for some storage systems, such convex cost structures may arise. Consider for example, ``indirect'' energy storage in the form of a demand response system \cite{newbery2018shifting}. The cost of shifting demand (per unit of energy) is likely to be increasing as function of quantity to be shifted since one needs to call upon new types of consumers, for whom delaying consumption may be more expensive. As a second example, consider indirect storage in the form of a ``vehicle to grid'' system. Different electric vehicle owners may ask for different prices to allow their vehicle to be used for energy storage; thus for storing more energy one needs to call upon vehicles with higher per-unit storage rates. We conclude that our stylized linear-quadratic storage model may represent some types of storage systems, and hence can be used to understand, at least qualitatively, the impact of storage on price formation. 

\begin{remark}[Representative storage agent] \normalfont\label{remark_param_agent}
To understand how to choose the parameter values for a ``representative'' storage agent in this model, recall that 
the term $\frac{\beta}{2}\left(Q-\overline Q_0\right)^2 $ models the cost, per hour, of ``renting'' $Q$ \text{MWh} of electricity from the storage system, while the term $\alpha q^2$ models the additional cost one has to pay for discharging at rate $q$. Since we aim to model a storage system operating in daily cycles, we consider a stylized daily cycle shown in Figure \ref{soc.fig}, and choose the parameters $\alpha$ and $\beta$ so that each component of the costs corresponds to $50\%$ of the typical levelized cost of storage of 1\text{MWh} of electricity in current energy markets \cite{schmidt2019projecting}.

\begin{figure}
\centerline{\includegraphics[width=0.6\textwidth]{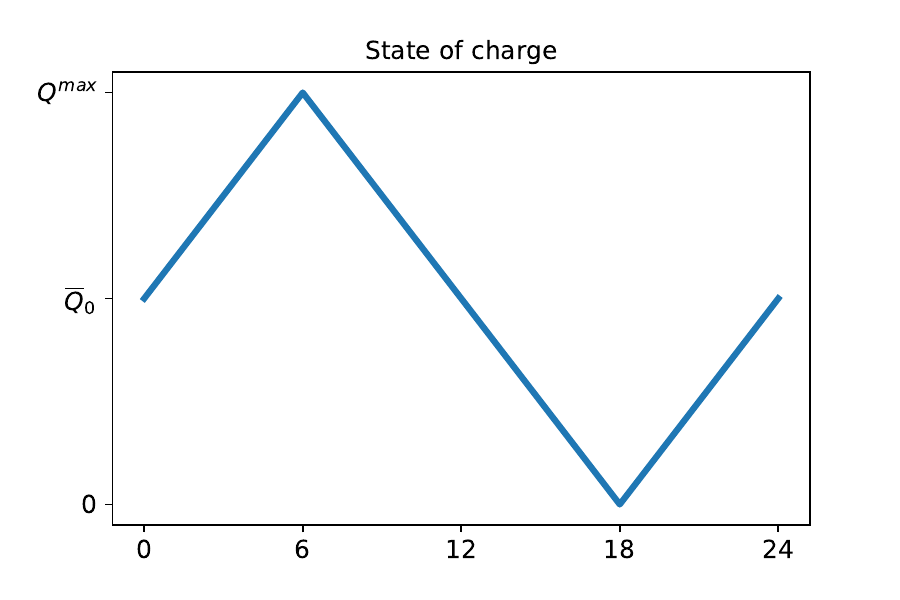}}
\caption{Stylized daily cycle for a storage facility.}
\label{soc.fig}
\end{figure}

Letting $Q^{max} = 1$ \text{MWh} and $\overline Q_0 = 0.5$ \text{MWh}, we easily find that the rental cost of our 1\text{MWh} storage system over 24 hours equals exactly $\beta$ euros, while the discharge cost equals $\frac{\alpha}{12}$ euros.
With a typical round-trip efficiency of $20\%$ and considering a typical intraday price differential of around $70$ \euro{}/MWh, we find a levelized cost of around $14$ \euro{}. This yields $\beta = 7$ and $\alpha = 84$. The choice of $\gamma$ is less important since it is a penalty parameter to enforce the terminal constraint; it should be taken large enough to make sure that $Q_T$ is close to $\overline Q_0$. In the sequel, to describe a ``representative'' storage agent with 1MWh capacity, we therefore take $\beta = 7$, $\alpha = 84$ and $\gamma = 500$. 
\end{remark}



The following theorem determines the optimal strategy of the storage owner.

\begin{theorem}\label{main}Assume that the price process $(P_t)$  belongs to $\mathcal{H}^2$.
    Define the following auxiliary functions :

    \begin{equation}
        f(t,T):=\frac{1-ue^{-2\omega\left(T-t\right)}}{1+ue^{-2\omega\left(T-t\right)}}.
\end{equation}
\begin{equation}
    f_1(t,s,T):=\omega\frac{e^{-\omega(s-t)}-ue^{-\omega(2T-s-t)}}{1+ue^{-2\omega(T-t)}},
    \end{equation}
    with $u:=\frac{\sqrt{\a \b}-\g}{\sqrt{\a \b}+\g}$ and $\omega = \sqrt{\frac{\beta}{\alpha}}$.

    There exists an unique  optimal control for the control problem \eqref{Optimal}. Furthermore, the optimal strategy of a storage agent with state equation (\ref{state}) solving the optimal control problem (\ref{Optimal}) admits the following closed-loop representation : 
    \begin{equation} \label{optimalsol}
    q^\star_t=\omega f(t,T)\left(Q^{q^\star}_t-\overline{Q}_0\right) - \E\left[ \int_t^T f_1(t,s,T)\left(\frac{P_s}{2\a}-\kappa_s\right)ds \Big| \mathcal{F}_t\right]  + \frac{P_t}{2\a}
\end{equation}
where $(Q^{q^\star}_t)$ is the solution of 
\begin{equation*}
    dQ^{q^\star}_t=(-q_t^\star+\kappa_t)dt + \rho_t dW_t + \int_{\mathbb{R}^{\star}} \Xi(t,e)\widetilde N(dt,de), \,\, Q^{ q^\star}_0= \overline{Q}_0.
\end{equation*}

\end{theorem}


Theorem \ref{main} with an exogenous price process is a rather standard result in the literature, proofs of its variants can be found in several papers.
Our treatment largely follows \cite{kohlmann_global_2002}, but since this reference does not allow for a general filtration, we provide a proof in Appendix \ref{proofmain.sec} for completeness.

We close this section with the following corollary, which shows that identical storage agents can be aggregated into a representative agent. The assumption of independence of $W^1,\dots,W^p$ is imposed to simplify notation and can be easily relaxed to allow for correlation between individual Brownian motions. 
\begin{corollary}\label{representative}Let $W^1,\dots,W^p$ be independent $\mathbb F$-Brownian motions, and consider $p$ storage agents without external supply and without jump terms in their dynamics, that is,
$$
dQ^{i,q^i}_t = -q^i_t dt + \rho^i_t dW^i_t,\quad Q^{i,q^i}_0 = \overline Q^i_0. 
$$
Assume that these agents have identical parameters (i.e., $\alpha_i=\alpha$, $\beta_i=\beta$ and $\gamma_i=\gamma$ for all $i=1,\ldots,p$).  Then, the optimal injection or withdrawal of this group of agents is equal to the optimal injection / withdrawal of a single agent with parameters $(\alpha/p,\beta/p,\gamma/p,\bar \rho)$ with $\bar \rho_t = \sqrt{\sum_{i=1}^p (\rho^i_t)^2}$ and initial storage level $\overline Q_0 = \sum_{i=1}^p \overline Q^i_0$. 
\end{corollary}
\begin{proof}
By Theorem \ref{main}, the optimal strategies  of agents are given by
\begin{align*}    {q}^{i,\star}_t=\omega f(t,T)\left({Q}^{i,{q}^{i,\star}}_t-\overline{Q}_0^{i}\right) - \E\left[ \int_t^T f_1(t,s,T)\frac{P_s}{2\a}ds | \mathcal{F}_t\right]  + \frac{P_t}{2\a}, \,\, i=1,\ldots,p.
\end{align*}
Denoting $ q^\star_t:= \sum_i q^{i,\star}_{t}$ and similarly for other quantities and adding up the equations, we get:
\begin{align*}    
{q}^\star_t&=\omega f(t,T)\left({Q}^{{q}^\star}_t-\overline{Q}_0\right) - p\E\left[ \int_t^T f_1(t,s,T)\frac{P_s}{2\a}ds | \mathcal{F}_t\right]  + p\frac{P_t}{2\a}\\
d{Q}^{{q}^\star}_t&=-{q}^\star_t dt + \bar\rho_t  dW_t,  \quad {Q}^{{q}^\star}_0= \overline{Q}_0,
\end{align*}
where $W_t :=  \int_0^t \frac{1}{\bar \rho_t}\sum_{i=1}^p \rho^i_t d W^{i}_t$ is a $\mathbb F$-Brownian motion by Levy's theorem. Notice that the functions $f$ and $f_1$ do not change when $\alpha$, $\beta$ and $\gamma$ are multiplied by the same constant and do not depend on $\rho$. This coincides with the optimal strategy of a single storage agent with parameters $(\alpha/p,\beta/p,\gamma/p,\bar\rho )$.

\end{proof}


\section{Price formation}\label{price.sec}

Given the optimal strategy of a storage agent obtained in the previous section for an exogenous price process, we now aim to determine the equilibrium electricity price for a specific demand level. We consider an electricity market comprising conventional producers, consumers, renewable producers and storage facilities. 
 
The conventional producers are assumed to have a deterministic net supply function for the intraday market $C: \R \longmapsto \R $, which depends solely on the real-time electricity price and represents the net residual supply, that is, $C(P)$ is the difference of the total production of the conventional producers and the quantity they deliver, at any given time, according to their day-ahead market position and other long-term contracts. A positive value of $C(P)$ corresponds to a positive net supply to the intraday market and a negative value corresponds to a positive net demand.

The consumers are assumed to have a price-independent stochastic real-time net demand process $\widetilde D_t$. The demand $\widetilde D_t$ can also be both positive and negative, and the decomposition of market agents into conventional producers and consumers is arbitrary: the function $C$ may include price-dependent demand of consumers and the process $\widetilde D$ may include stochastic supply of conventional producers. The only assumption is that the stochastic component of demand/supply does not depend on the price. 

The renewable producers offer their full intermittent capacity $R_t$. The \textit{residual demand} is defined by ${D}_t:=\widetilde D_t-R_t$ and represents the energy requirement that must be fulfilled by conventional producers and storage players.

We assume that in the market there are  $n$ storage agents with different parameters and characteristics. In view of Corollary \ref{representative}, each storage agent can be seen as an aggregate of several agents with identical characteristics. 

We make the realistic assumption that the electricity price is constrained within a range $[\underline{P},\overline{P}]$. As a result, there are situations where the energy demand exceeds the producers' capacity, leading to a failure to satisfy the supply-demand equation.

{\color{black}In view of the above discussion, we shall now define the market equilibrium. In the definition, we use the following notation:
    \begin{align*}
        f^j(t,T)&:=\frac{1-u^je^{-2\omega^j\left(T-t\right)}}{1+u^je^{-2\omega^j\left(T-t\right)}},\\
    f_1^j(t,s,T)&:=\omega^j\frac{e^{-\omega^j(s-t)}-u^je^{-\omega^j(2T-s-t)}}{1+u^je^{-2\omega^j(T-t)}},
    \end{align*}
    with $u^j:=\frac{\sqrt{\a^j \b^j}-\g^j}{\sqrt{\a^j \b^j}+\g^j}$ and $\omega^j = \sqrt{\frac{\beta^j}{\alpha^j}}$.
    
\begin{definition}[Market equilibrium]\label{eq.def}
A market equilibrium is a price process $(P_t)_{0\leq t\leq T}$ and a collection of strategies of storage agents $(q^j_t)_{0\leq t\leq T}^{j=1,\dots,n}$, such that, for all $t\in[0,T]$,
\begin{align*}
P_t&=\essinf\left\{ \chi \in L^2(\mathcal{F}_t) : {D}_t\leq C(\chi)+\sum_j \tilde q_t^j(\chi) \right\} \wedge \overline{P} \vee \underline{P} ,\,\, 0 \leq t \leq T,
\end{align*}
where for each $t\in [0,T]$ and each $j = 1,\dots,n$ the function $\tilde q^j_t$ is defined by 
\begin{align*}
\tilde q_t^j(\chi)&=\omega^j f^j(t,T)\left(Q_t^{j,q^j}-\overline{Q}_0^j\right) - \E\left[ \int_t^Tf^j_1(t,s,T)\left(\frac{P_s}{2\a^j}-\kappa^{j}_s\right)ds | \mathcal{F}_t \right] + \frac{\chi}{2\a^j},
\end{align*}
with $Q^{j,q^j}$ the solution of the following SDE:
$$
dQ^{j,q^j}_t=(-q_t^j+\kappa^j_t)dt + \rho^j_t dW_t + \int_{\mathbb{R}^{\star}} \Xi^j(t,e)\widetilde N(dt,de), \,\, Q^{j,q^j}_0= \overline{Q}^j_0,
$$
where 
$$
q^j_t=\omega^j f^j(t,T)\left(Q^{j,q^j}_t-\overline{Q}^j_0\right) - \E\left[ \int_t^T f^j_1(t,s,T)\left(\frac{P_s}{2\a^j}-\kappa^j_s\right)ds \Big| \mathcal{F}_t\right]  + \frac{P_t}{2\a^j}.
$$
\end{definition}
}



In the following sections we discuss the existence and computation of the equilibrium price process. We first focus on a deterministic toy model, and then turn to the stochastic case.

\subsection{A fully deterministic toy model}\label{toy.sec}

We will now concentrate on a simplified model, enabling us to perform explicit computations and analyze limiting cases. In this section, we make the following assumptions:

\begin{assumption} \label{assdetermin}${}$
    \begin{itemize}
    \normalfont
        \item[(i)] There is one unique type of storage in the electricity market, i.e. $n=1$.
        \item[(ii)] The storage agents are not impacted by the Brownian and Poisson noises, i.e. $\rho \equiv 0$, $\kappa=0$  and $\Xi \equiv 0$. 
        \item[(iii)] The residual demand $D_t$ is deterministic.
        \item[(iv)] The conventional supply function $C$ is  linear, i.e. there exist $C_0,C_1>0$ such that $C(P)=C_0+C_1P$. 
        \item[(v)] The price is unbounded, i.e. $\overline{P}=\infty$ and $\underline{P} = -\infty$.
        \item[(vi)]$Q_0=0$.
    \end{itemize}
\end{assumption}
In particular, Assumption \ref{assdetermin} (v)  implies that energy demand is always fulfilled.
In this toy model, the definition of the price process simplifies to the following system of equations: for all $t$,
\begin{equation} 
\begin{cases}
  {D}_t = C(P_t)+ q_t,  \\ 
 \displaystyle  q_t=\omega f(t,T)Q^q_t - \int_t^Tf_1(t,T;s)\frac{P_s}{2\a}ds + \frac{P_t}{2\a}.
\end{cases}
\end{equation}


In the case of an exogenous deterministic price process $(P_t)$, we first provide in the following proposition an explicit expression of the optimal strategy in the present framework. Its proof can be found in Appendix  \ref{proof3.1}.
\begin{proposition}\label{price_exogeneous}
Let Assumption \ref{assdetermin} hold true and let $P \in L^1([0,T])$. Define 
 \begin{equation}\label{represc1}
     c_1(P):= \frac{\gamma\int_0^T \cosh(\omega(T-s)) \frac{P_s}{2\alpha} ds+ \beta\int_0^T \frac{\sinh(\omega(T-s))}{\omega} \frac{P_s}{2\alpha} ds}{\gamma  \omega\sinh\left(\omega T\right)+\beta \cosh\left(\omega T\right)}.
 \end{equation}
 Then, for all $t\in [0,T]$, the optimal strategy and the corresponding state process are 
\begin{align}
    q_t &= -c_1(P) \omega^2\cosh\left(\omega t\right) + \frac{P_t}{2\alpha}+\int_0^t \omega\sinh(\omega(t-s)) \frac{P_s}{2\alpha} ds,\\
    Q_t^q &= -c_1 \omega\sinh\left(\omega t\right) + \int_0^t \cosh(\omega(t-s))  \frac{P_s}{2\alpha}ds.
\end{align}    
\end{proposition}

The following theorem provides the explicit expression for the equilibrium price. Its proof is provided in Appendix \ref{prooftoy.sec}.
\begin{theorem}\label{toy.thm}
   Let $\Tilde\omega:=\omega\sqrt{\frac{C_1}{C_1 + \frac{1}{2\alpha}}}$ and $\Tilde{c}_1$ be given by
   \begin{equation}\label{c1}
       \Tilde{c}_1 := \frac{\gamma A' + \beta A}{2\alpha\omega(\gamma  \omega\sinh (\omega T)+\beta \cosh (\omega T))-\gamma B'-\beta B},
\end{equation}
where
\begin{align} \label{c_1_constantsA}
    A &:= \frac{\omega}{\tilde\omega(C+1/2\alpha)}\int_0^T ( D_s - C_0) \sinh(\tilde\omega(T-s)) ds,\\ \label{c_1_constantsB}
B&: = \frac{\omega^3}{\tilde\omega(C+1/2\alpha)}\int_0^T \cosh(\omega s) \sinh(\tilde\omega(T-s)) ds,
\\
    A'&: =  \frac{\omega}{(C+1/2\alpha)}\int_0^T (D_s - C_0) \cosh(\tilde\omega(T-s)) ds, \label{c_1_constantsA'}\\
    B'&: = \frac{\omega^3}{(C+1/2\alpha)}\int_0^T \cosh(\omega s) \cosh(\tilde\omega(T-s)) ds.\label{c_1_constantsB'}
\end{align}
   Moreover, let us define the following auxiliary function:
    \begin{multline}
        X(t):= \frac{\omega}{\tilde\omega(C_1+\frac{1}{2\alpha})}\int_0^t\left(D_s - C_0\right) \sinh(\tilde\omega(t-s)) ds \\+  \Tilde{c}_1\frac{\omega^3}{\tilde\omega(C+\frac{1}{2\alpha})}\int_0^t \cosh(\omega s) \sinh(\tilde\omega(t-s)) ds
    \end{multline}
    Then, for all $t \in [0,T]$, the  unique equilibrium price is given by
\begin{equation}
        P_t= \frac{D_t - C_0 - \frac{\omega}{2\alpha} X(t) + \Tilde{c}_1 \omega^2 \cosh(\omega t)}{C_1+ \frac{1}{2\alpha}}.
    \end{equation}
\end{theorem}


In the following proposition, we study the behavior of the equilibrium price when the amount of storage in the market is either very small or very large. To this end, in view of Corollary \ref{representative}, we consider a storage agent with parameters $(\alpha/p,\beta/p,\gamma/p)$ and interpret $p\to 0$ as the ``low storage'' limit and $p\to \infty$ as the ``high storage'' limit. We find, as expected, that in the low storage limit, the price converges to the price without storage. In the high storage limit, we find that the derivative of the price converges to zero, in other words, the storage assets completely cancel price fluctuations. 

\begin{proposition}
        Let  $\alpha,\beta,\gamma>0$. Denote by $(P^p_t)$ the equilibrium electricity price in a market with a single storage agent with parameters $(\alpha/p,\beta/p,\gamma/p)$. Then, for all $t\in [0,T]$,
    \begin{equation}
        \lim_{p\downarrow 0} P^p_t = \frac{D_t-C_0}{C_1} .
    \end{equation}
Assume in addition that the residual demand function $t \mapsto  D_t$ is differentiable. Then, for all $t\in [0,T]$,
    \begin{equation}
        \lim_{p\to \infty}\frac{\partial P^p_t}{\partial t}= 0.
    \end{equation}
\end{proposition}

\begin{proof} \emph{Part 1.} To make the dependence on $p$ explicit, we denote the quantities appearing in Theorem \ref{toy.thm}, computed with parameters $(\alpha/p,\beta/p,\gamma/p)$, by $\tilde\omega^p$, $\tilde c^p_1$, $A^p$, $B^p$, $A^{p\prime}$, $B^{p\prime}$ and $X^p(t)$. Note that $\omega$ does not depend on $p$. Clearly, $\lim_{p\downarrow 0} \tilde\omega^p = \omega$. This implies that $A^p$, $B^p$, $A^{p\prime}$, $B^{p\prime}$ converge to finite limits as $p\downarrow 0$, and $X^p(t)$ remains uniformly bounded. On the other hand,
$$
 \tilde{c}^p_1 = \frac{p(\gamma A^{p\prime} + \beta A^p)}{2\alpha\omega(\gamma \omega\sinh (\omega T)+\beta \cosh (\omega T))-\gamma p B^{p\prime}-\beta p B^p},
$$
so that $\lim_{p\downarrow 0} \tilde c^p_1 = 0$. Substituting this into the expression for $P^p_t$, we finally find:
$$
\lim_{p\downarrow 0}P^p_t= \lim_{p\downarrow 0}\frac{D_t - C_0 - \frac{\omega p}{2\alpha} X^p(t) + \Tilde{c}^p_1 \omega^2 \cosh(\omega t)}{C_1+ \frac{p}{2\alpha}} = \frac{D_t - C_0}{C_1}.
$$

\noindent \emph{Part 2.} 
Direct differentiation of $P^p_t$ yields: 
\begin{align*}
    \frac{\partial P^p_t}{\partial t}  
    &= \underbrace{\frac{1}{C_1+p/2\a}\left\{D'_t-\frac{\b/p}{2C_1\frac{\a^2}{p^2}+\frac{\a}{p}}\int_0^t(D_s-C_0)\cosh(\omega\sqrt{\frac{C_1}{C_1+\frac{p}{2\a}}} (t-s))ds\right\}}_{:=E(p)} \\ &\quad+\underbrace{c_1^{p}\frac{\b\sqrt{\b}}{\a\sqrt{\a}(C_1+p/2\a)^{3/2}} \sinh(\omega\sqrt{\frac{C_1}{C_1+\frac{p}{2\a}}} t))}_{:=F(p)},
\end{align*}
with $(D_t')$ being the time derivative of $ D_t$.
Clearly, $\underset{p \rightarrow + \infty}{\lim}E(p)= 0$. Examining the expressions (\ref{c_1_constantsA}--\ref{c_1_constantsB'}), we see that, as $p\to \infty$, $pA^p$, $pB^p$, $pA^{p\prime}$ and $pB^{p\prime}$ converge to nonzero limits, which means that $\tilde c^p_1$ remains bounded and therefore $\underset{p \rightarrow + \infty}{\lim}F(p)= 0$, which implies $$\underset{p \rightarrow + \infty}{\lim}\frac{\partial P^p_t}{\partial t}= 0.$$

\end{proof}
\subsection{Price Formation in a stochastic framework} \label{section_price_sto}
In this section we suppose that the filtration $\mathbb F$ is the completed natural filtration of \textcolor{black}{a $2n+1$-dimensional Brownian motion $W$ and a Poisson random measure $N$  defined on $ \R ^\star \times [0,T]$  with compensator $\nu(de)dt$, such that $\nu(de)$ is a $\sigma$-finite measure on $\R ^\star$, equipped with its Borel field $\mathcal{B}(\R ^\star)$. Recall that we denote by $\widetilde{N}$ the compensated jump measure, i.e. $\widetilde{N}(dt, de) := N(dt, de) - \nu(de)dt$.}




For a given control $q^j \in \mathcal{A}$, the state process of the $j$-th agent satisfies the following SDE:
\begin{align*}
                dQ^{j,q^j}_t&=(-q^j_t + \kappa^j_t) dt+ \textcolor{black}{\r_t^j dW_t} + \int_{\R^\star} {\Xi^j(t,e)} \textcolor{black}{\widetilde N(dt,de)}, \, \, Q^{j,q^j}_0=\overline{Q}_0^j,
\end{align*}
We make the following assumption on the coefficients of this process, the residual demand and the supply function:
\begin{assumption} \label{state_process}${}$
\begin{itemize}
\item[i.] For all $j=1,\ldots,n$, $\rho^j\in \textcolor{black}{\mathcal H^2_{2n+1}}$, $\Xi^j \in \mathcal H^2_\nu$  and  there exists a non-negative function $\Psi_{j}$ such that $\int_{\mathbb{R}^\star}\Psi_{j}^2(de)\nu(de)<\infty$ such that $|\Xi^{j}(t,e)| \leq \Psi^{j}(e)$  for all $t \in [0,T]$ and $e \in \R^\star$. 
\item[ii] For all $j=1,\ldots,n$, the external energy source $\kappa^j_t$ follows the SDE 
\begin{equation}
     \textcolor{black}{d\kappa_t^j= \iota^j(t,\kappa^{j}_t)dt + \vartheta^j(t,\kappa^{j}_t)dW_t}, \quad \kappa^j_0= 0,
\end{equation}
where $\iota^j:\Omega \times [0,T] \times \R \mapsto \R$ and \textcolor{black}{$\vartheta^j:\Omega \times [0,T] \times \R \mapsto \R^{2n+1}$} are such that for all $x \in \R$, the processes $\iota^j(\cdot,\cdot,x)$ and $\vartheta^j(\cdot,\cdot,x)$ are $\mathbb{F}$-progressively measurable and there exists $K<\infty$ such that, for all $t\in[0,T]$, and all $(x,x') \in \R^2$, 
\begin{align*}
    &|\iota^j(t,x)|+ \textcolor{black}{\|\vartheta^j(t,x)\|} \leq K(1+|x|), \\
    & |\iota^j(t,x) - \iota^j(t,x')|  +\textcolor{black}{\|\vartheta^j(t,x) - \vartheta^j(t,x')\|}  \leq K|x-x'|.
\end{align*}
 
\item[iii.] The residual demand $({D}_t)$ has the following dynamics :
\begin{equation} \label{Demand}
    d{D}_t= \mu(t,{D_t}) dt+\sigma(t,{D_t}) \textcolor{black}{dW_t},
\end{equation}
where the coefficients $\mu:\Omega \times [0,T] \times \R \longmapsto \R$ and $\sigma:\Omega \times [0,T] \times \R \longmapsto \R^{2n+1}$ are such that for all $x \in \R$, the processes $\mu(\cdot,\cdot,x)$ and $\sigma(\cdot,\cdot,x)$ are $\mathbb{F}$-progressively measurable. Furthermore, suppose that there exists $K<\infty$ such that for all $t\in[0,T]$, and all $(x,x') \in \R^2$, 
$$|\mu(t,x)|+ \textcolor{black}{\|\sigma(t,x)\|} \leq K(1+|x|),\quad |\mu(t,x) - \mu(t,x')|+\textcolor{black}{\|\sigma(t,x) - \sigma(t,x')\|}  \leq K|x-x'|
$$
\item[iv.] The function $L:\R \to \R,\ x \mapsto C(x)+x\sum_i\frac{1}{2\a^i}$ is invertible and its inverse is uniformly Lipshitz continuous with constant $C$. 
\end{itemize}
 \end{assumption}

We assume that every storage agent in the market is a price taker (considers the price process as exogenous). The $j$-th agent determines its injection / withdrawal strategy by solving the optimization problem
\begin{equation*}
            \underset{q^j \in \A }{\inf} \E \left\{ \int_0^{T} - \left(P_tq^j_t-\frac{\alpha^j}{2}\left(q^j_t\right)^2\right) +\frac{\beta^j}{2}\left(Q^{j,q^j}_t-\overline{Q}_0^j\right)^2 dt + \frac{\gamma^j}{2}\left(Q^{j,q^j}_T-\overline{Q}_0^j\right)^2 \right\}.
\end{equation*}
By Theorem \ref{main}, we obtain that there exists an unique optimal control, which is given, for each $j=1,\ldots,n$, by 
\begin{align*}
q^{j,\star}_t&=\omega^jf^j(t,T)\left(Q_t^{j,q^{j,\star}}-\overline{Q}_0^{j}\right) - \E\left[ \int_t^Tf^j_1(t,T;s)\left(\frac{P_s}{2\a^j}-\kappa^{j}_s\right)ds | \mathcal{F}_t \right] + \frac{P_t}{2\a^j},
\end{align*}
and the controlled state process of the $j$-th agent  along the optimal control is given by
\begin{align*}
                dQ^{j,q^{j,\star}}_t&=(-q^{j,\star}_t + \kappa^j_t) dt+ \r_t^j \textcolor{black}{dW_t} + \int_{\R^\star} \Xi^j(t,e) \textcolor{black}{\widetilde N(dt,de)}, \, \, Q^{j,q^{j,\star}}_0=\overline{Q}_0^j.
\end{align*}



 We first give the following preliminary Lemma, which, for a given price process $({P}_t)$, gives a representation of $(q_t^j)$, for all $1 \leq j \leq n$, in terms of the solution of a specific FBSDE. 
 \begin{lemma}\label{repres}
 Let $({P}_t) \in \mathcal{H}^2$ be a given price process. For $1 \leq j \leq n$, let $(Y_t^{j,1}, Z_t^{j,1},\Gamma_t^{j,1})$ and $(Y_t^{j,2}, Z_t^{j,2}, \Gamma_t^{j,2})$ be the unique solutions in  $\mathcal{S}^2 \times \mathcal{H}_{2n+1}^2 \times\textcolor{black}{ \mathcal{H}^2_{\nu}}$ of the following BSDEs:
 \begin{align}
 \begin{cases}
dY_t^{j,1}&=-e^{-\omega^j t}\left(\frac{P_t}{2\a^j}-\kappa^j_t\right) dt + Z^{j,1}_tdW_t+\int_{\mathbb{R}^\star}\Gamma^{j,1}(t,e)\Tilde{N}(dt,de), \,\,\, Y_T^{j,1}=0, \nonumber\\
dY_t^{j,2}&=-e^{\omega^j t}\left(\frac{P_t}{2\a^j}-\kappa^j_t\right)dt + Z^{j,2}_tdW_t+\int_{\mathbb{R}^\star}\Gamma^{j,2}(t,e)\Tilde{N}(dt,de), \,\,\, Y_T^{j,2}=0.\\
\end{cases}
 \end{align}
 Then, for $1 \leq j \leq n$, the optimal control $(q_t^{j,\star})$ is given by:
 \begin{align}\label{rep1}
 q^{j,\star}_t=\omega^{j} f^j(t,T) \left(Q_t^{j,q^{j,\star}}-\overline{Q}_0^j\right) +\frac{P_t}{2\a^j}-F_1^j(t,T)Y^{j,1}_t +F_2^j(t,T)Y^{j,2}_t, 
 \end{align}
where
\begin{equation}\label{coef}
\begin{cases}
     F_1^j(t,T):=\omega^j \frac{e^{\omega^j t}}{1+u^je^{-2\omega^j (T-t)}}\\
     F_2^j(t,T):=u^j\omega^j \frac{e^{-\omega^j (2T-t)}}{1+u^je^{-2\omega^j (T-t)}},
\end{cases}
  \end{equation}
with $u^j:=\frac{\sqrt{\a^j \b^j}-\g^j}{\sqrt{\a^j \b^j}+\g^j}$.
 \end{lemma}
\begin{proof}
Using \eqref{coef}, we can rewrite the expression of $(q_t^{j,\star})$ as follows:
\begin{align} \label{equation_q}
    q^{j,\star}_t&=\omega^j f^j(t,T)\left(Q_t^{j,q^{j,\star}}-\overline{Q}_0^j\right) - F_1^j(t,T)\E\left[ \int_t^T e^{-\omega^j s}\left(\frac{P_s}{2\a^j}-\kappa^j_s\right)ds | \mathcal{F}_t \right] \nonumber \\&
       + F_2^j(t,T)\E\left[ \int_t^T e^{\omega^j s}\left(\frac{P_s}{2\a^j}-\kappa^j_s\right)ds | \mathcal{F}_t \right]+ \frac{P_t}{2\a^j}.
\end{align}
Taking conditional expectation with respect to $\mathcal{F}_t$ in the dynamics of $Y_t^{j,1}$ and $Y_t^{j,2}$, we get that
$$Y_t^{j,1}= \E\left[ \int_t^T e^{-\omega^j s}\left(\frac{P_s}{2\a^j}-\kappa^j_s\right)ds | \mathcal{F}_t \right]$$
and
$$Y_t^{j, 2}= \E\left[ \int_t^T e^{\omega^j s}\left(\frac{P_s}{2\a^j}-\kappa^j_s\right)ds | \mathcal{F}_t \right].$$
Therefore, by combining the above expressions with \eqref{equation_q}, we derive $\eqref{rep1}$.
\end{proof}

We now present two alternative existence results for the equilibrium price: the first result requires the coefficients to be deterministic functions and does not allow for the presence of a jump term, but gives existence for arbitrary time horizon, and the second result does not require this additional assumption, but provides existence only on a small interval. 

\begin{theorem} \label{theorem_price}
    Suppose Assumption \ref{state_process} is in force and the following additional conditions hold true:
\begin{itemize}
\item[(i)] The coefficients $\mu$, $\sigma$, $\vartheta^j$, $\iota^j$, $\rho^j$, for $j \in \{1,\ldots,n\}$, are deterministic maps.
\textcolor{black}{\item[(ii)] For $t\in [0,T]$ and $x\in \mathbb R$, let $G(t,x)$ be a real-valued square matrix of size $2n+1$, whose lines are defined as follows: 
\begin{align*}
    G_{i}(t,x)&:=\rho^i_t, \quad &\text{for } 1 \leq i\leq n, \\
    G_{n+1}(t,x)&:=\sigma(t,x) &\\
    G_{n+1+i}(t,x)&:= \vartheta^i(t,x) \quad &\text{for }1 \leq i\leq n.
\end{align*}
There exists $\lambda>0$ such that for all $t \in [0,T]$, $x \in \R$, and for all $u\in \mathbb R^{2n+1}$ with $u\neq 0$, 
$$
u^\top G(t,x)G(t,x)^T u >\lambda\|u\|^2.
$$}
\item[(iii)] For all $j \in \{1,\ldots,n\}$, $\Xi^{j} \equiv 0$.
\end{itemize}    
Then there exists an unique equilibrium price process defined on the interval $[0,T]$.
\end{theorem}

\begin{proof}
The first step is to establish the equivalence between the existence/uniqueness of a {\color{black}market equilibrium in the sense of Definition \ref{eq.def}} and the existence/uniqueness of a solution of a specific coupled FBSDE system. 

{\color{black}
Assume that there exists a market equilibrium price $(P_t)_{0\leq t\leq T} \in \mathcal H^2$ and strategies $(q^j_t)_{0\leq t\leq T}^{j=1,\dots,n}$ which satisfy Definition \ref{eq.def}. By Lemma \ref{repres}, for all $j=1,\dots,n$, 
\begin{align*} 
\begin{cases}
    q^j_t&=\omega^j f^j(t,T) (Q_t^{j,q^j}-\overline{Q}_0^j) +\frac{P_t}{2\a^j}-F_1^j(t,T)Y^{j,1}_t +F_2^j(t,T)Y^{j,2}_t,\,\,  t\in [0,T],  \\
    dY_t^{j,1}&=-e^{-\omega^j t}(\frac{P_t}{2\a^j}-\kappa^j_t) dt + Z^{j,1}_tdW_t,\,\,\, Y_T^{j,1}=0\\
    dY_t^{j,2}&=-e^{\omega^jt}(\frac{P_t}{2\a^j}-\kappa^j_t)dt + Z^{j,2}_tdW_t, \,\,\, Y_T^{j,2}=0.\\
\end{cases}
\end{align*}
}
Observe now that, {\color{black}using Definition \ref{eq.def} and Assumption \ref{state_process}.iv, the equilibrium price} $(P_t)$ can be written, for all $t\in [0,T]$, as 
$$P_t=  L^{-1}\left({D}_t -  \sum_{j=1}^n \omega^j f^j(t,T) (Q_t^{j,q^j}-\overline{Q}_0^j) -F_1^j(t,T)Y^{j,1}_t +F_2^j(t,T)Y^{j,2}_t\right)\wedge \overline{P} \vee \underline{P}.$$
We therefore deduce that there exists a solution of the following FBSDE system:
\begin{align} \label{FBSDE_sys1}
\begin{cases}
        d{D}_t&= \mu(t,{D_t}) dt+\sigma(t,{D_t}) \textcolor{black}{dW_t} \nonumber  \\
      dQ^{j,q^j}_t&= \left[-\left(\omega^jf^j(t,T)(Q_t^{j,q^j}-\overline{Q}_0^j) -F_1^j(t,T)Y^{j,1}_t +F_2^j(t,T)Y^{j,2}_t \right. \right.\notag\\&\qquad\left. \left. + \frac{h(t,(Q_t^j)_j,(\kappa^j_t)_j, {D}_t,(Y^{j,1}_t)_j, (Y^{j,2}_t)_j)}{2\a^j}\right)+\kappa^j_t \right]dt + \r^j_t \textcolor{black}{dW_t} \nonumber \\
    d\kappa_t^j&= \iota^j(t,\kappa_t^j)dt + \vartheta^j(t,\kappa_t^j)\textcolor{black}{dW_t}, \quad \kappa^j_0= 0 \nonumber \\
       dY_t^{j,1}&=-e^{-\omega^jt}\left(\frac{h(t,(Q_t^j)_j,(\kappa^j_t)_j, {D}_t,(Y^{j,1}_t)_j, (Y^{j,2}_t)_j)}{2\a^j}-\kappa_t^j\right) dt + Z^{j,1}_tdW_t \nonumber \\
         dY_t^{j,2}&=-e^{\omega^j t}\left(\frac{h(t,(Q_t^j)_j,(\kappa^j_t)_j, {D}_t,(Y^{j,1}_t)_j, (Y^{j,2}_t)_j)}{2\a^j}-\kappa_t^j\right) dt + Z^{j,2}_tdW_t \nonumber  \\
   Q^{j,q^j}_0&=\overline{Q}^j_0,\,\, Y_T^{j,1}=Y_T^{j,2}=0,
\end{cases}
\end{align}
where 
\begin{align}
&h(t,q^1,\ldots,q^n,\kappa^1,\ldots,\kappa^n,d,y^{1,1},\ldots,y^{1,n},y^{2,1},\ldots,y^{2,n}) \nonumber \\ &:=L^{-1}\left(d-\sum_{j=1}^n \omega^j f^j(t,T) (q^j-\overline{Q}_0^j) -F_1^j(t,T)y^{j,1} +F_2^j(t,T)y^{j,2}\right)\wedge \overline{P} \vee \underline P.
\end{align}

{\color{black}Conversely, if there exists a solution to the fully coupled FBSDE system (\ref{FBSDE_sys1}), we can observe that $P_t:=h(t,(Q_t^{j,q^j})_j,(\kappa^j_t)_j, D_t,(Y^{j,1}_t)_j, (Y^{j,2}_t)_j)$ and the corresponding strategies satisfy Definition \ref{eq.def}. Therefore, we conclude that the existence and uniqueness of a market equilibrium is equivalent to the existence and uniqueness of the solution of the FBSDE system (\ref{FBSDE_sys1}).}

Under Assumption \ref{state_process} and the additional assumptions of this theorem, we can apply  Theorem 2.29 of \cite{carmona2016lectures}  from which we deduce the existence and uniqueness of the solution of the FBSDE system (\ref{FBSDE_sys1}). In view of the above arguments, we can conclude that there exists an unique equilibrium price process defined on the interval $[0,T]$.

\end{proof}

We provide below an alternative existence and uniqueness result, which does not require additional assumptions, for a small time interval.

\begin{proposition}
     Suppose Assumption \ref{state_process} is in force. There exists a constant $\delta>0$  such that there exists an  unique market equilibrium on the interval $[0,\delta]$.
\end{proposition}
The proof  follows by the same arguments used in the proof of Theorem \ref{theorem_price} combined with Theorem 3.2 of \cite{li2014lp} (which gives the existence and uniqueness of the solution of a fully-coupled FBSDE with jumps in a small time interval).


\section{Short-term price impact of storage} \label{shortterm.sec}




In this section, we illustrate the theoretical findings of this paper with numerical simulations. We first consider the deterministic toy model of Section \ref{toy.sec} and then the stochastic model of Section \ref{section_price_sto}.
\textcolor{black}{Both models are calibrated} using French data from ENTSOE\footnote{https://transparency.entsoe.eu/}.
 


Throughout this section, we assume that there is no external production providing energy to the storage directly and the supply function of the conventional producers is linear: $C(P) = C_0 + C_1 P$, where $C_0$ represents the baseline net supply.

The energy demand is considered entirely exogenous, so it will remain the same across different storage scenarios. Furthermore, the plotted variable is not the energy demand itself but the residual demand, that is the demand minus the renewable production.
Finally, to align with Remark \ref{remark_param_agent}, the time horizon is always considered to be $24$ hours.

\paragraph{Calibration of $C(P)$}
{Recall that the function $C(P) = C_0 + C_1 P$ corresponds to the net residual supply of conventional producers to the intraday market (defined as the difference between the total supply of conventional producers and the day-ahead forecast of this quantity). We make the simplifying assumption that the consumers do not operate in the intraday market.
$C(P)$ is calibrated in two steps. Firstly, the intercept $C_0$ is taken equal to the minimal value of the intraday residual demand during the estimation period (January 1st, 2024 to March 30th, 2024), yielding $C_0=-7546$MW.  The negative sign of this value is due to an excess of renewable energy production beyond initial projections.}

{Secondly, for calibrating the coefficient $C_1$, we use merit order curve data obtained from the French day-ahead electricity market.
Figure \ref{bidoffer.fig} shows the aggregated bid and offer curves. The volume corresponding to the offers to the left of the intersection between the two curves is sold in the day-ahead market. We then conjecture that the volume corresponding to the offers to the right of the intersection will be offered in the intraday market of the corresponding hour. We see that the intersection is located within a quasi-linear section of the offer curve, and this is also the case for other dates and hours not shown here. We estimate the slope of the linear section of the curve for all days and all hours from January 1st, 2016 to October 5, 2017.\footnote{A different estimation period is used for reasons of data availability.} The mean value of the coefficient is $151.77$ MWh/Eur with a standard deviation of  $108.03$ MWh/Eur.}

\begin{figure}[H]
\centerline{\includegraphics[width=\textwidth]{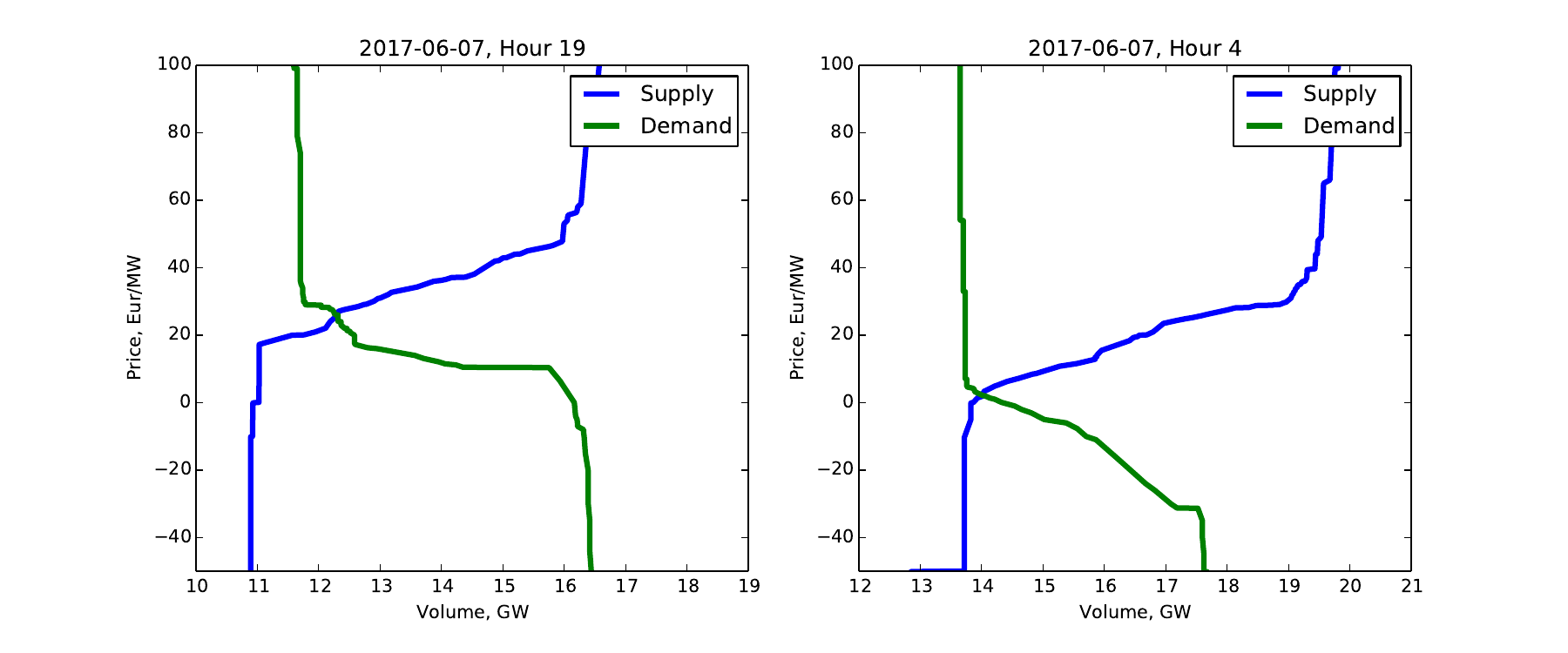}}
\caption{Aggregated bid and offer curves in the French intraday electricity market for a peak hour (left) and off-peak hour (right). Data source: EPEX Spot.}
\label{bidoffer.fig}
\end{figure}

\paragraph{Deterministic framework}
We begin by placing ourselves within the deterministic framework outlined in Section \ref{toy.sec}.  In this toy setting, we have an explicit expression of both the electricity price and the storage agent's strategy at each time point. The objective here is to understand the role of storage in a context where there is no impact due to the randomness of renewable production or demand. {\color{black} 
We model the excess intraday demand (difference between the actual demand and the day-ahead demand forecast), denoted by $D_t$, using a periodic function of the following form:
\begin{equation*}
  D_t=\theta_1 \sin(\theta t)+ D_0.
\end{equation*}
We use this form for the intraday demand with two artificial peak/off-peak periods in order to highlight the potential impact of storage on price variations.
The parameters are calibrated using the data from the French intraday market from January 1st, 2024 to March 30th, 2024: $ D_0=1500$MW is the average of the minimum and the maximum value of residual demand over this period, $\theta_1=6862.5$MW is equal to one half of the difference between the maximum and the minimum value of residual demand, and we fix $\theta=\pi/6$ (hours$^{-1}$) to include 2 peak/off-peak periods within a day.} 

{\color{black}Figure \ref{fig:deter_3_scenarios} illustrates the injection rate and the price impact of a 10GWh storage (\textcolor{black}{$p=10000$} identical agents having parameters $\alpha=84,\beta=7, \gamma=500$ as described in Corollary \ref{representative}). The left graph plots the residual demand and the injection / withdrawal rate, which exhibits a periodic  pattern, similar to that of the residual demand but with a smaller amplitude.  During periods of high energy demand, the storage agent sells surplus energy to compensate the residual demand, while adopting the reverse strategy during off-peak times by storing excess energy for future use. This can be seen as an arbitrage strategy, wherein energy is stored during low-price periods and sold during high-price periods. The right graph illustrates the impact of energy supply of storage assets on the market price: the presence of storage reduces the differences between peak and off-peak prices: the greater the amount of storage introduced into the market, the more narrow the price interval becomes, leading to reduced peak prices and higher off-peak prices. }



\begin{figure}[htbp] 
    \centering 
    \centerline{Electricity demand and withdrawal rate \hspace{2.7cm} Electricity price \hspace{0.7cm}${}$}
    \includegraphics[width=\textwidth]{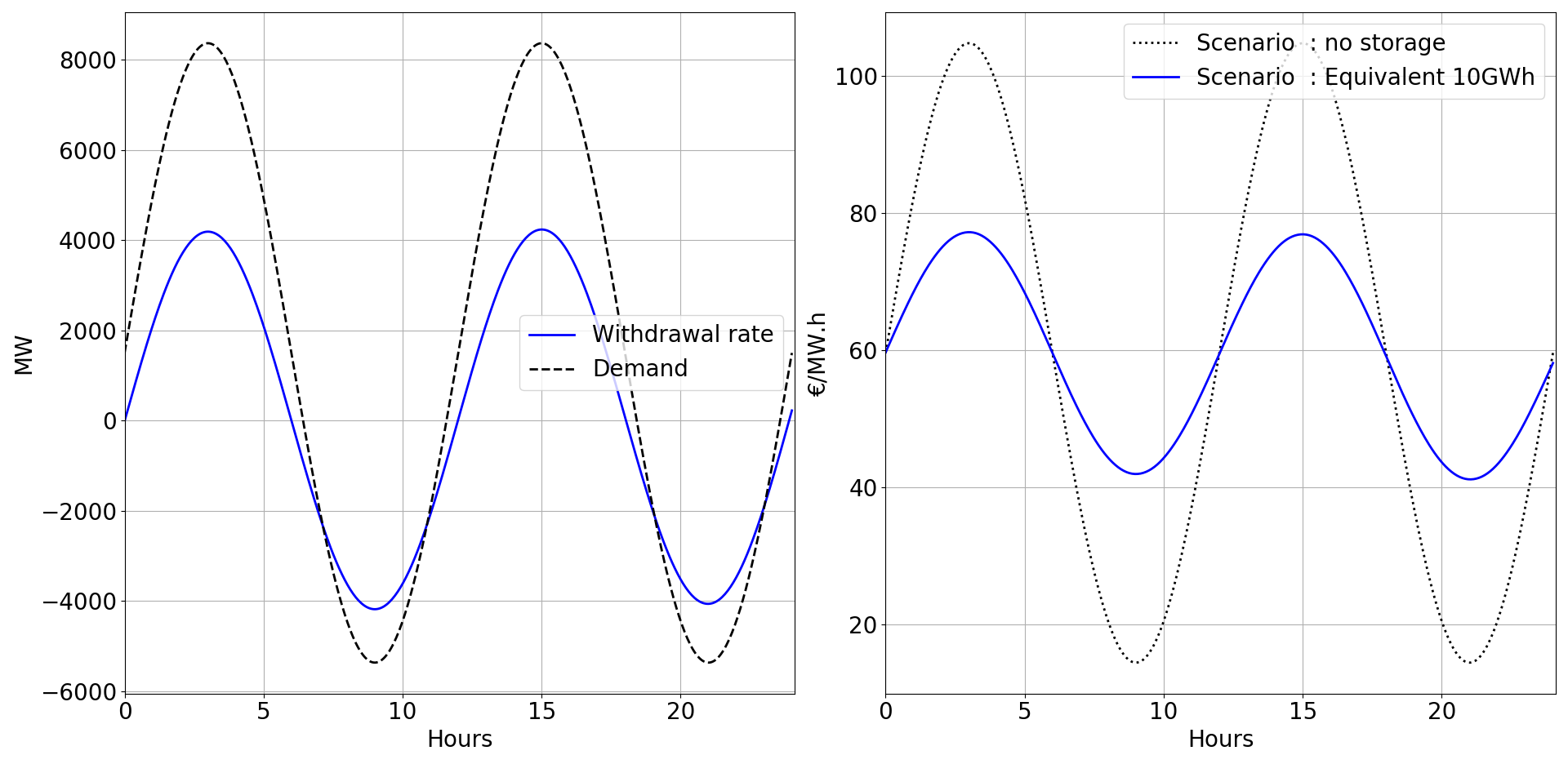}
    \caption{Left: intraday demand and the corresponding storage withdrawal rate for a 10GWh storage;  Right: corresponding electricity prices. }
    \label{fig:deter_3_scenarios}
\end{figure}

\paragraph{Stochastic framework} 
In the stochastic framework, from the proof of Theorem \ref{theorem_price}, the price process is related to the unique solution of a forward-backward stochastic differential equation. To solve this equation, we shall use the numerical scheme of \cite{bender_time_2008}. Below is a brief overview of the scheme adapted to our case. Introduce the following FBSDE on $[0,T]$:
\begin{equation}    
\begin{cases}
 \displaystyle  X_t=x+ \int_0^tb(s,X_s,Y_s)ds+\int_0^t \sigma(s,X_s) dW_s \\
 \displaystyle    Y_t=\int_t^T f(s,X_s,Y_s)ds - \int_t^T Z_s dW_s
\end{cases}
\end{equation}
 with $b,\s $ and $ f$ deterministic functions satisfying appropriate assumptions (see e.g. \cite{carmona2016lectures} and \cite{bender_time_2008}) ensuring that the above FBSDE admits a unique solution on $[0,T]$. In this case, the solution of this FBSDE can be characterised through a decoupling field, a deterministic function $u$ such that $X$ and $Y$ are connected through the following formula : 
 \begin{equation*}
     Y_t=u(t,X_t).
 \end{equation*}
 The numerical scheme for a time interval divided into $n$ segments can be expressed as follows, with $u_i^{n,0}=0$, at the $m^{th}$ iteration: 
 \begin{equation}
     \begin{cases}
        X_0^{n,m}=x \\ 
        X_{i+1}^{n,m}= X_{i}^{n,m} + b(t_i, X_{i}^{n,m},u_i^{n,m-1}( X_{i}^{n,m}))h + \s (t_i, X_{i}^{n,m}) \Delta W_{i+1},\\
        Y^{n,m}_n=0, \\
        Z^{n,m}_i=\frac{1}{h} \E[Y_{i+1}^{n,m} \Delta W_{i+1} | \mathcal{F}_{t_i}] \\
        Y_i^{n,m} = \E[Y_{i+1}^{n,m}+ f(t_i,X_i^{n,m},Y_{i+1}^{n,m})h| \mathcal{F}_{t_i}] \\
        u_i^{n,m}(X_i^{n,m})=Y_i^{n,m}
        
     \end{cases}
 \end{equation}
 with $h=\frac{T}{n}$.
The conditional expectations are then computed with least squares regression.
In this section, we will use $\Lambda=50000$ simulated trajectories and \textcolor{black}{$n=480$} time steps.
In this stochastic setting, we model the intraday residual demand $(D_t)$, defined as the difference between the actual demand and the day-ahead demand forecast, and intraday renewable production $(R_t)$, defined similarly as the difference between the actual renewable production and the day-ahead renewable production forecast, by independent Ornstein-Uhlenbeck processes following the procedure in Appendix \ref{calibration}. This estimation is done in the period from January 1st to March 30th, 2024 and  yields the following estimates: 

\begin{align*}
        &\mu^{\Tilde D}=8.09 \text{MW}, &&\theta^{\Tilde D}=316 h^{-1}, &&\sigma^{\Tilde D}=35615 \text{MW}/h^{1/2}, \\
        &\mu^R=7.58\text{MW}, &&\theta^R=176 h^{-1}, &&\sigma^R= 26219\text{MW}/h^{1/2}.
\end{align*}


Similarly to our approach in the deterministic toy model, we examine the influence of changing the quantity of storage in the market on the electricity price dynamics. 
We consider two different scenarios: scenario $1$ with no storage, compared to scenario $2$ with $p=10000$ identical agents with parameters $\alpha=84,\beta=7,\gamma=500,\rho=0.01$ representing an equivalent storage capacity of 10 GWh as explained in Remark \ref{remark_param_agent}.
{\color{black}Figure \ref{fig:sto_2}, left graph shows a simulated trajectory of residual demand together with the corresponding injection / withdrawal trajectory, and the right graph shows the intraday price evolution with and without storage. In line with the deterministic case, we observe a reduction of price differences, where the abrupt price fluctuations due  to shocks in residual demand are smoothed by increased storage capacity.} 

\begin{figure}[H] 
    \centering
\centerline{Electricity demand and withdrawal rate \hspace{2.7cm} Electricity price \hspace{0.7cm}${}$}
    
    \includegraphics[width=\textwidth]{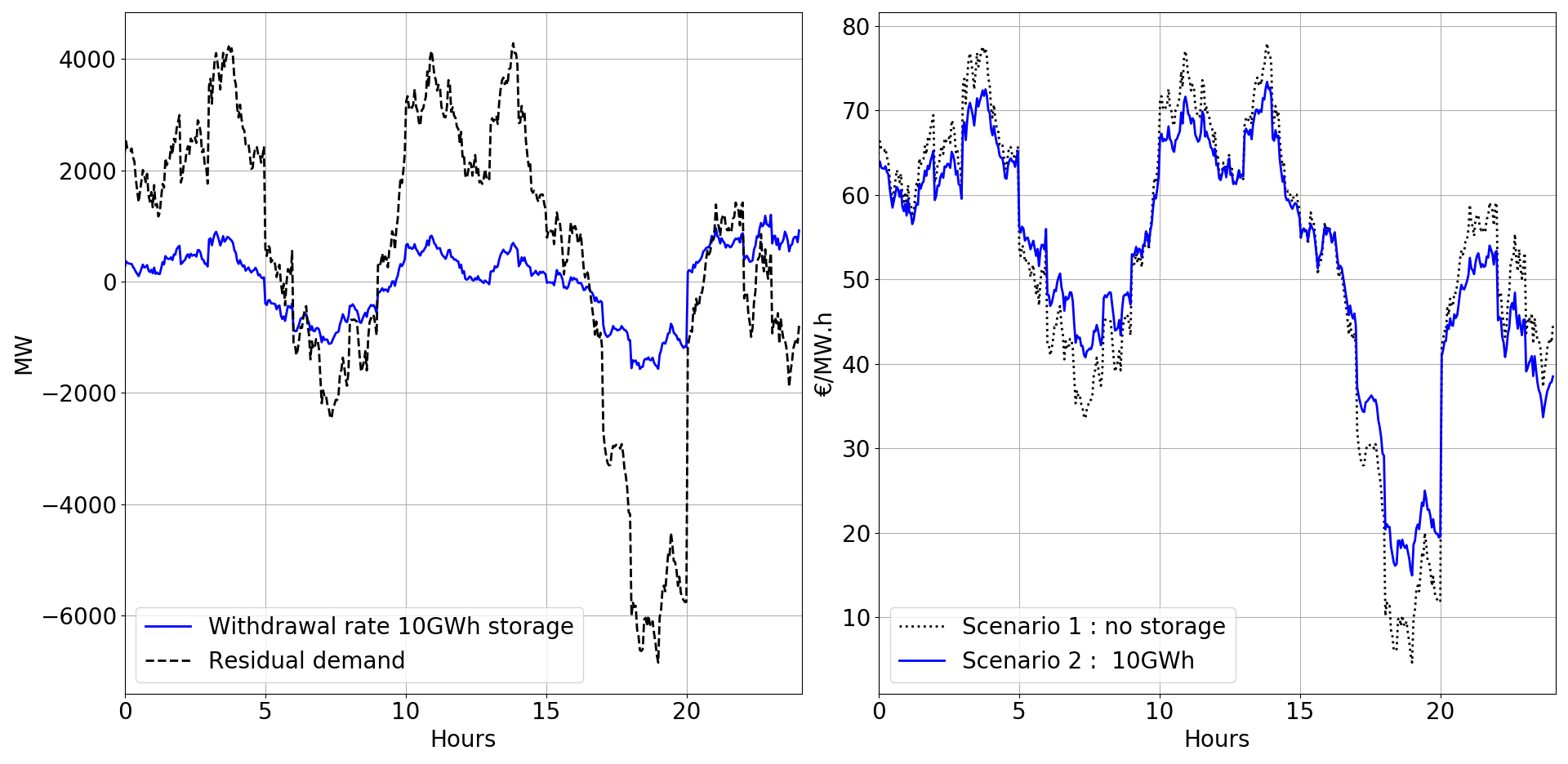}
    \caption{On the left, intraday residual demand with the withdrawal rate of a 10GWh storage. On the right, electricity prices in two scenarios with or without storage}
    \label{fig:sto_2}
\end{figure}


To characterize the impact of storage on electricity price volatility, we consider the same representative agent as in the previous paragraph and plot in Figure \ref{fig:vol} the evolution of the average volatility {(computed over 50000 simulations)} depending on the number of identical storage agents in the market.
A strong decreasing trend is visible, showing that \textcolor{black}{storage agents} stabilize market prices not only in terms of reducing price differences but also in terms of lowering price volatility. 

\begin{figure}[H] 
    \centering
    \includegraphics[width=\textwidth]{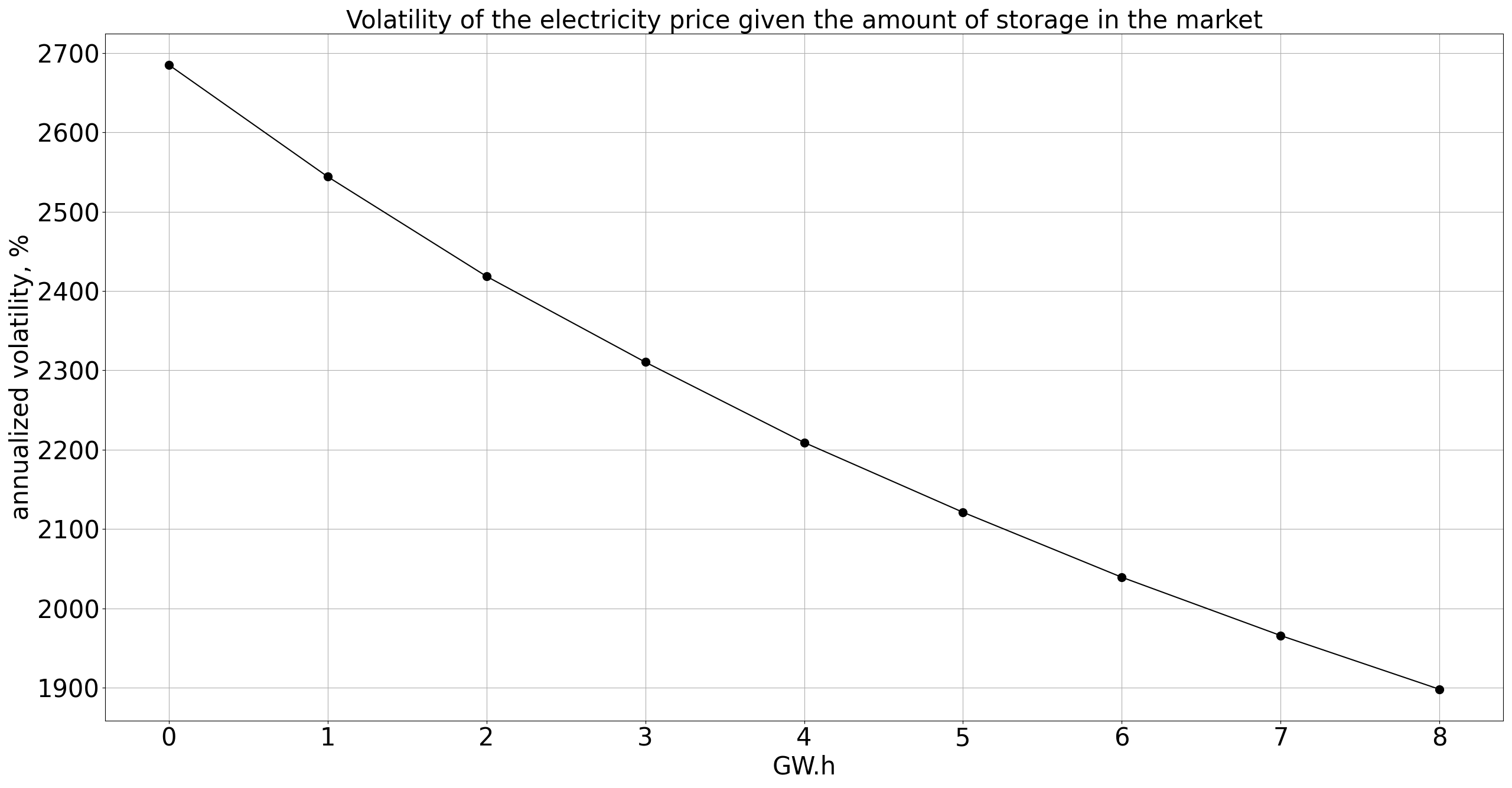}
    \caption{Evolution of the volatility given the amount of storage in the market}
    \label{fig:vol}
\end{figure}

\section{Long-term impact and profitability of storage}\label{longterm.sec}
A potential obstacle to the increasing deployment of energy storage is the ``cannibalization effect'': as new storage agents enter the market, their price impact which, as we have seen, leads to lower price differences and lower price volatility, will reduce their own revenues. 
\begin{figure}[H]
    \centering
    \includegraphics[width= \textwidth]{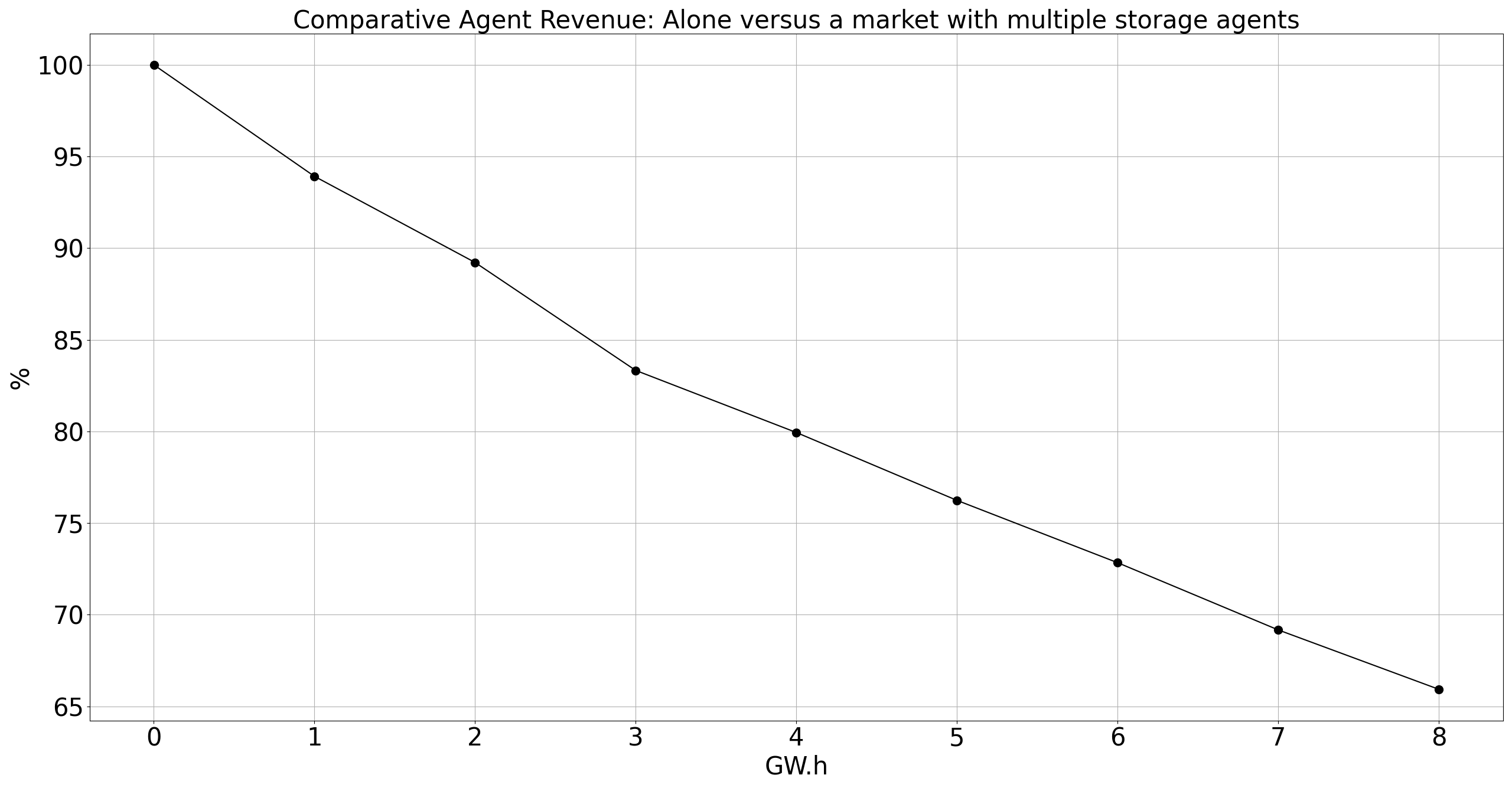}
    \caption{Ratio of the average daily revenue of an agent (1\text{MWh}) in the presence of additional storage to the average daily revenue in absence of additional storage, as function of the additional storage capacity.}
    \label{fig:agent_revenue}
\end{figure}

{\color{black} Figure \ref{fig:agent_revenue} plots the ratio of the average daily revenue (over 50000 trajectories, each representing one day) of a single storage agent  with 1\text{MWh} capacity in the presence of additional storage, to the average daily revenue of the same agent in absence of additional storage, as function of the additional storage capacity. As expected, increasing storage penetration leads to a considerable drop in daily revenues due to a reduction in price differences and overall arbitrage potential. This potential drop in revenues can discourage investment into storage capacity. However, this negative impact can be counter-balanced by a positive impact of increasing
penetration of intermittent renewable energies, which leads to higher price volatility, higher price
differences and potentially higher revenues for storage agents.}

{\color{black}To understand whether
the gains for storage agents related to increased renewable penetration will compensate
the losses due to their price impact in the long term, we consider the reference scenarios
for the evolution of the French electricity sector published by RTE, the French electricity
network operator.} In 2020, RTE published 6 reference scenarios\footnote{See \href{https://rte-futursenergetiques2050.com/}{rte-futursenergetiques2050.com}}  describing possible pathways for France towards carbon neutrality by 2050. These scenarios describe the evolution of production capacity (nuclear, renewable etc.) and provide benchmark values for 2030, 2040, and 2050. Each scenario is then tailored to three possible trajectories of energy demand. In this paragraph, we examine the evolution of revenues for a storage agent in every scenario. We will focus on the trajectory of increasing electrical consumption as outlined by RTE's reference scenario. 
{Figure \ref{fig:RTE_narrative} describes the narratives and main assumptions of the six scenarios we study}. The storage capacity in each scenario, and the multiplication coefficients for the renewable production and the energy demand are given  in Appendix \ref{RTE_coef}. }

Our objective is to quantify the impact, for each scenario, of the increase in the volume of available storage on the revenue of an individual storage agent, assuming that renewable energy production grows according to the same scenario. To this end, we make four additional assumptions: 
\begin{itemize}
    \item An increase in renewable production (respectively the energy demand) in a given scenario leads to a proportional increase in the production (respectively the energy demand) available  in the intraday market.
    \item The values found in the calibration of $C(.)$ remain the same. This can be justified by the long-term presence of gas power plants for short-term production adjustments.
    \item We consider that $25\%$ of the installed storage capacity is dedicated to the intraday market.\footnote{According to EPEX Spot (\href{https://www.epexspot.com/en/news/all-time-high-volumes-growth-spot-markets-illustrate-trust-trading-participants}{epexspot.com/en/news/all-time-high-volumes-growth-spot-markets-illustrate-trust-trading-participants}), in 2023, a total of 717.8 TWh was traded on EPEX SPOT, out of which 175.7 TWh was traded on the Intraday segment. We assume that the fraction of storage assets operating in the intraday market is similar.}
    \item We assume a one-to-one correspondence between the nameplate capacities of storage assets in MW given by RTE and storage capacities in MWh.\footnote{The ratio between the two can vary in practice, but, for example, Hornsdale Power Reserve has nameplate capacity of 150 MW and storage capacity of 194 MWh (\href{https://hornsdalepowerreserve.com.au/}{hornsdalepowerreserve.com.au}), so our assumption seems reasonable.}
\end{itemize}

\begin{figure}[H]
    \centering
    \includegraphics[width=\textwidth]{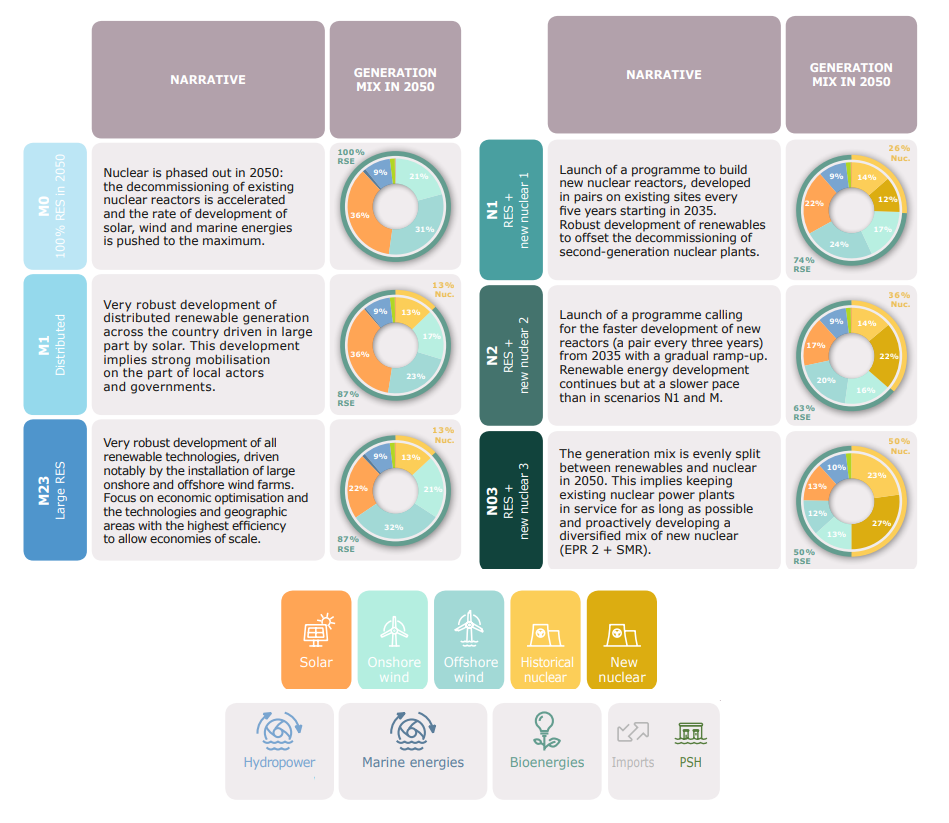}
    \caption{Brief description of every RTE scenario. Source: ``Energy Pathways to 2050'', \href{https://www.rte-france.com/analyses-tendances-et-prospectives/bilan-previsionnel-2050-futurs-energetiques}{rte-france.com/bilan-previsionnel-2050-futurs-energetiques}.}
    \label{fig:RTE_narrative}
\end{figure}

{{Figure \ref{fig:RTE_rev_quant} shows the quantiles of the  daily revenue of a 1MWh storage agent (computed over 50000 typical days) in each RTE scenario during the period from 2019 to 2050, and Figure \ref{fig:RTE_rev_mean} shows the mean revenue (over 50000 typical days) in each scenario during the same period}}.  We observe a continuous increase in mean revenues, so that a potential cannibalization effect is compensated by the  increase in the energy demand and the renewable production of the French energy mix. However, the interquantile range also increases in every scenario, highlighting the growing uncertainty of future revenues due to an increased proportion of renewable production. {To explore the long-term impact of storage on electricity prices, { Figure \ref{fig:RTE_vol_quant} shows the quantiles of electricity price volatility over 50000 days in the RTE scenarios, and Figure \ref{fig:RTE_vol_mean} shows the mean value (over 50000 typical days) of volatility over the same period in each scenario}. We observe that volatility increases due to the growth of renewable production, leading, for example, to a volatility peak in 2050 that is, on average, $3.7$ times higher than its value in 2019 in the M0 and M1 scenarios. In contrast, in the N0 scenario, the volatility is multiplied on average by $2.7$. This numerical application illustrates a drawback of scenarios with too many renewables and highlights one of the reasons why more storage is required as the share of renewables increases.}


\begin{figure}[H]
    \centering
    \includegraphics[width=\textwidth]{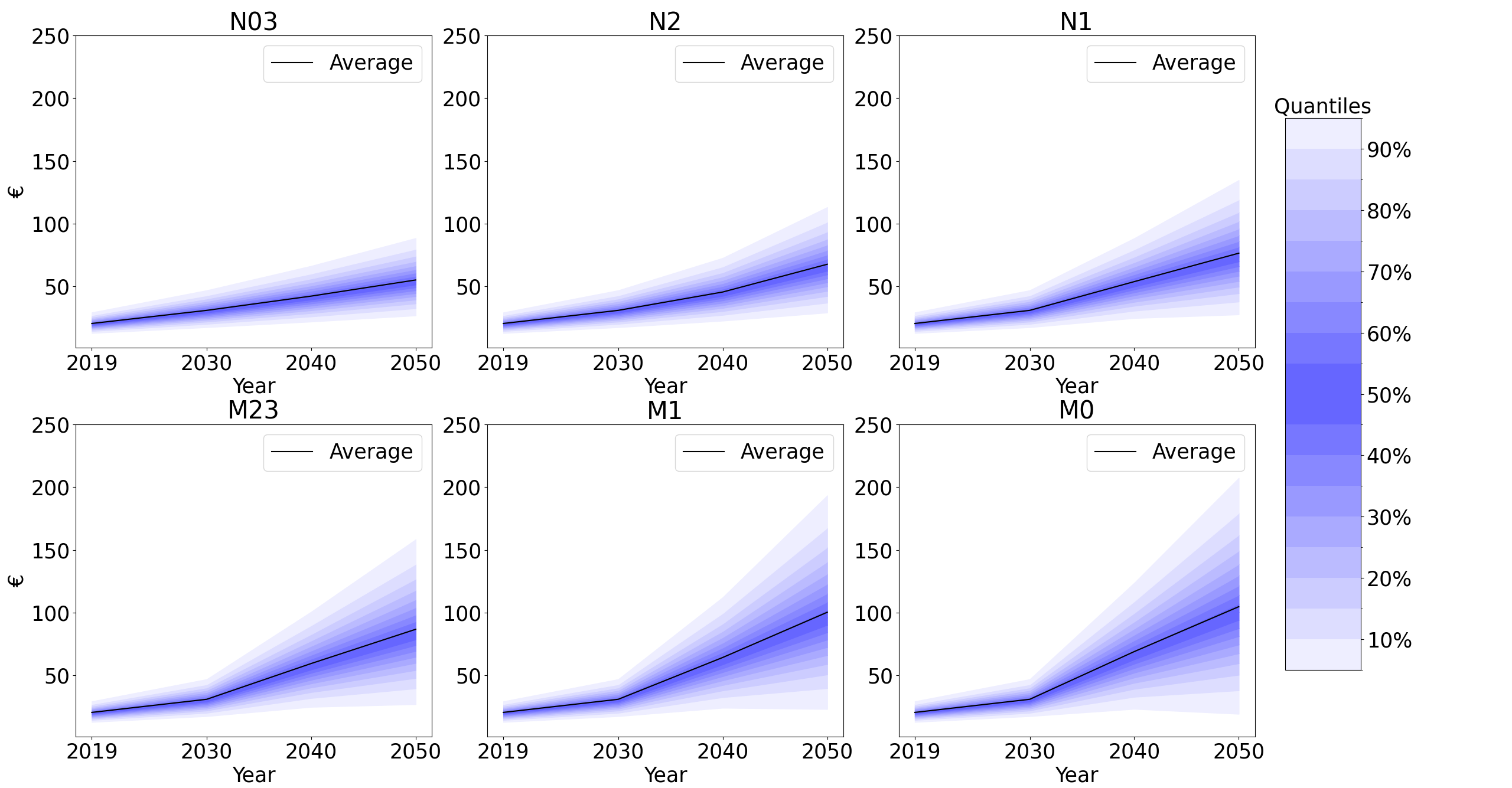}
    \caption{Evolution of the quantiles of the revenues of a 1MWh storage agent in the RTE scenarios}
    \label{fig:RTE_rev_quant}
\end{figure}

\begin{figure}[H]
    \centering
    \includegraphics[width=\textwidth]{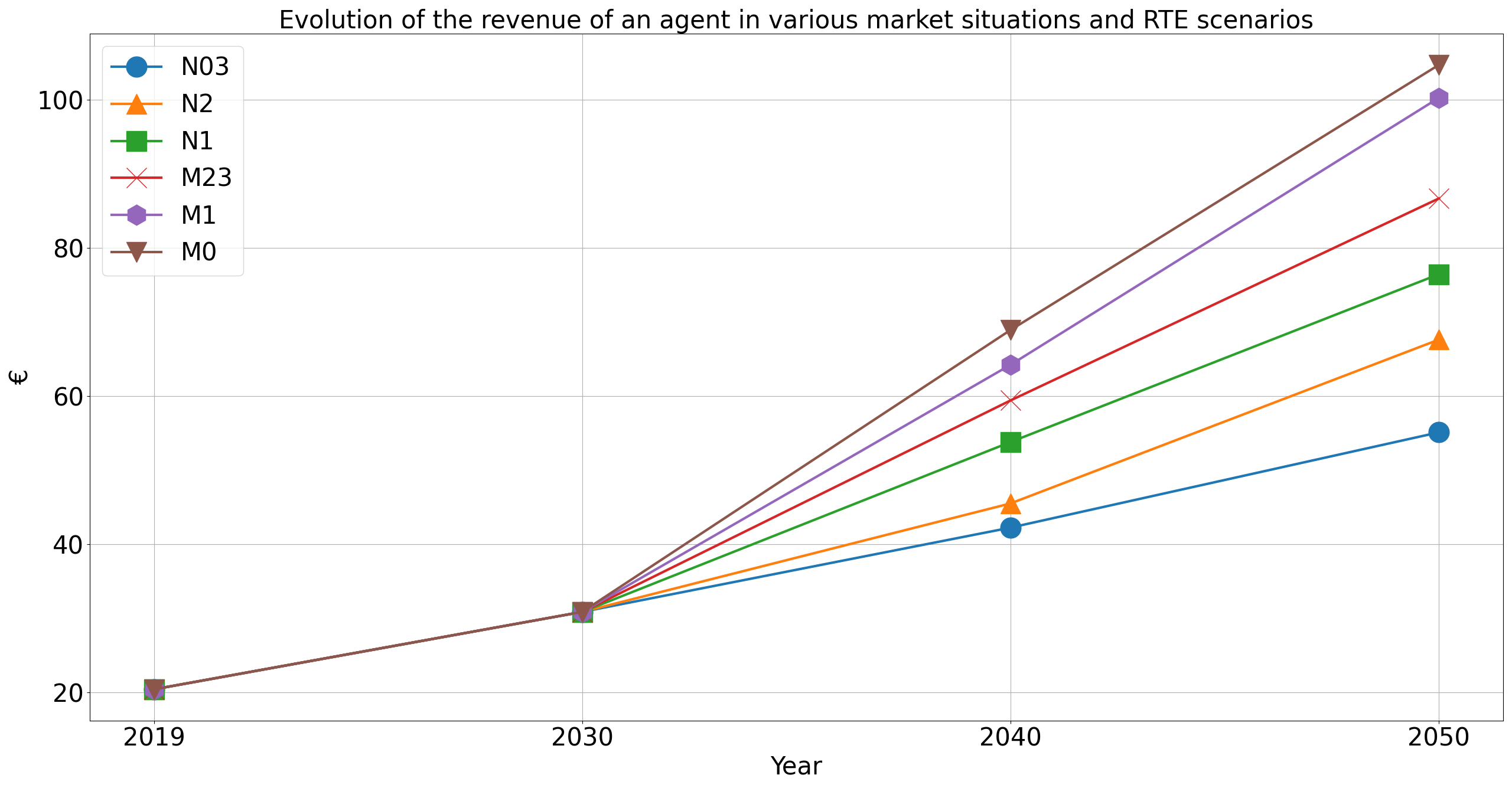}
    \caption{Evolution of the revenues of a 1MWh storage agent in the RTE scenarios}
    \label{fig:RTE_rev_mean}
\end{figure}

\begin{figure}[H]
    \centering
    \includegraphics[width=\textwidth]{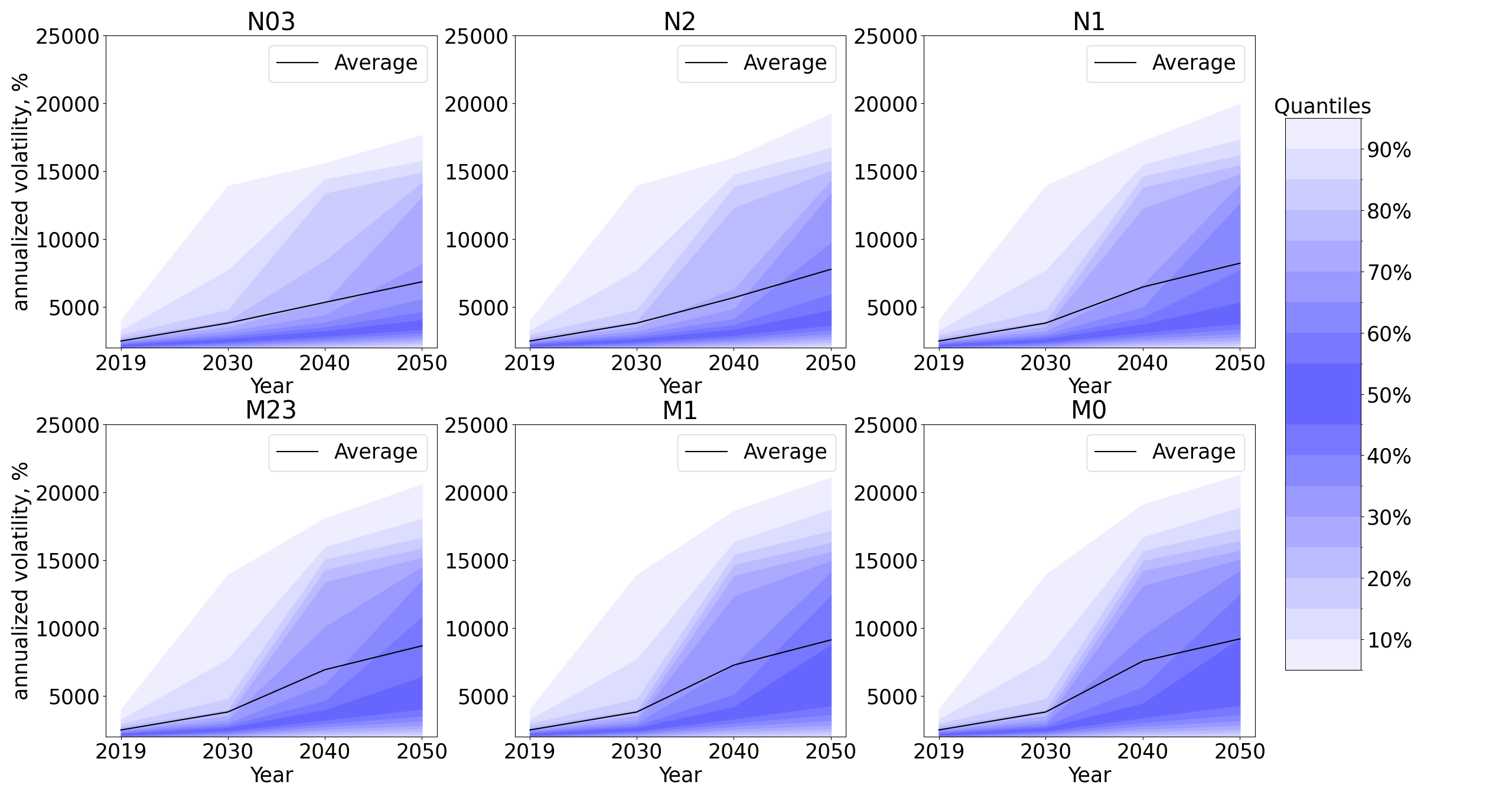}
    \caption{Evolution of the quantiles of electricity prices volatility in the RTE scenarios}
    \label{fig:RTE_vol_quant}
\end{figure}

\begin{figure}[H]
    \centering
    \includegraphics[width=\textwidth]{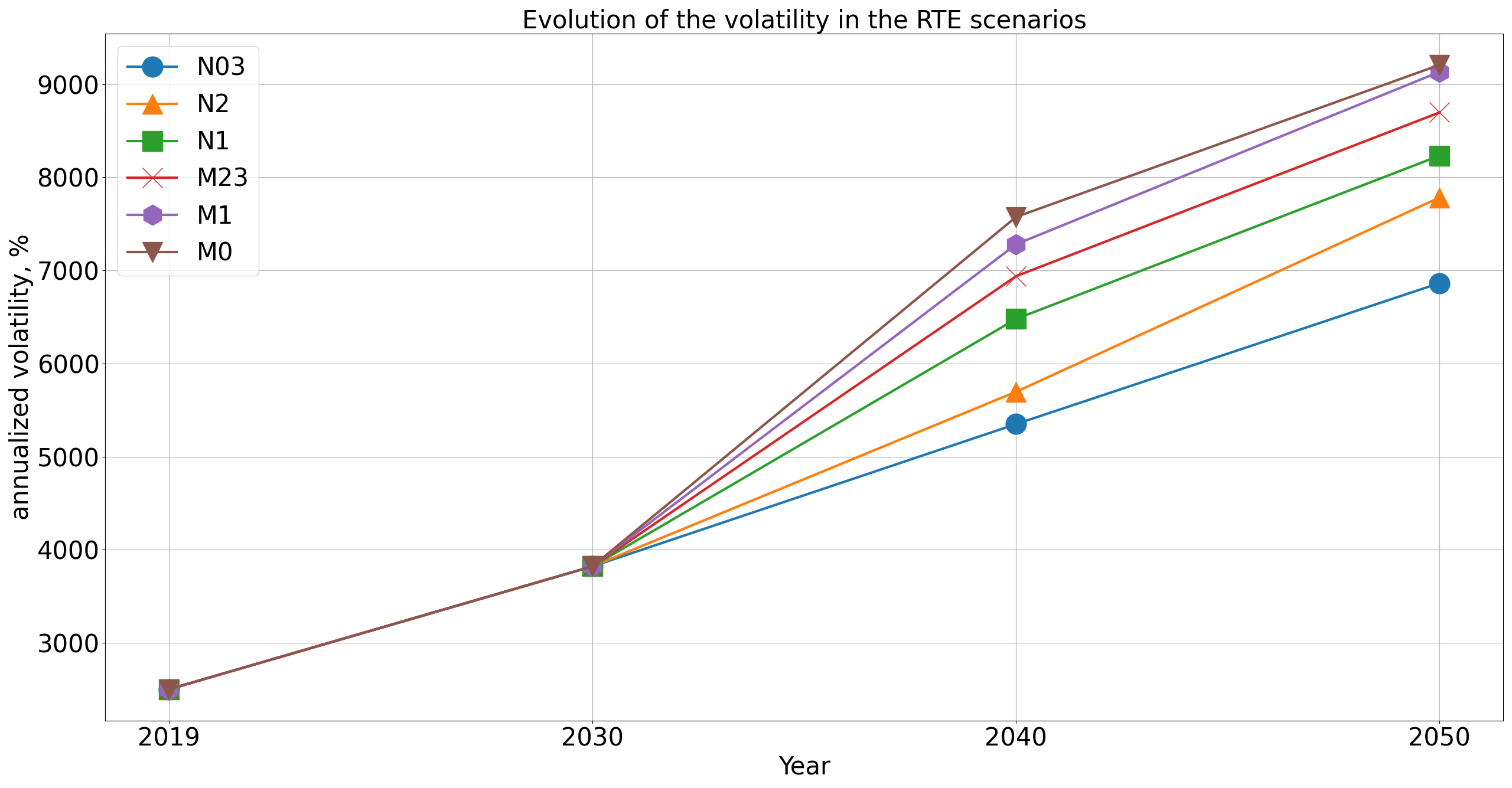}
    \caption{Evolution of the average electricity prices volatility in  RTE scenarios}
    \label{fig:RTE_vol_mean}
\end{figure}

\appendix

\section{Proof of Theorem \ref{main}}
\label{proofmain.sec}
The first step is to rewrite the optimization problem in a more convenient form. 
Observe first that the problem \eqref{Optimal} is equivalent to the following one:
\begin{align}\label{P1}
 \underset{q \in \mathcal{A}}{\inf} \mathbb{E}\left[\int_0^{T} \frac{\a}{2}\left(q_t-\frac{P_t}{2\a}\right)^2 +\frac{\beta}{2}\left(Q^q_t-\overline{Q}_0\right)^2 dt+\frac{\gamma}{2}\left(Q^q_T-\overline{Q}_0\right)^2\right]
\end{align}
For a given $q_t \in \mathcal{A}$, we define, for all $t \in [0,T]$,
\begin{equation}\label{change_var}
    a_t:=q_t-\frac{P_t}{2\a};\,\, dQ^a_t=\left(-a_t-\frac{P_t}{2\a}+\kappa_t\right)dt + \rho_t dW_t+ \int_{\mathbb{R}^{\star}} \Xi(t,e)\widetilde N(dt,de).
\end{equation}
Since $(P_t)\in \mathcal{H}^2$, the new control process $(a_t)$ also belongs to the set $\mathcal{A}$.
Problem \eqref{P1} can be then transformed into 
\begin{align}\label{opt}
\underset{a \in \mathcal{A}}{\inf}J(0,\overline Q_0,a) 
\end{align}
with
\begin{equation}
    J(0,\overline Q_0,a) :=\mathbb{E}\left[ \int_0^T\frac{\a}{2} a_t^2 dt + \int_0^T \frac{\beta}{2}\left(Q^a_t-\overline{Q}_0\right)^2 dt+\frac{\gamma}{2}\left(Q^a_T-\overline{Q}_0\right)^2\right].
\end{equation} 
By strict convexity and coercivity of the functional $J(0,\bar{Q}_0,\cdot)$, we deduce the existence and the uniqueness of a minimizer for \eqref{opt}.
To obtain the explicit representation of the optimal control, we first establish the following lemma.
\begin{lemma} \label{lemma}
 Let  $  a^\star \in \mathcal{A}$ be the optimal control minimizing $\eqref{opt}$. Then the following forward-backward stochastic differential equation with jumps (FBSDEJ) admits an unique solution 
  $(Y^{a^\star},Z^{a^\star},\Gamma^{a^\star},H^{a^\star}) \in \mathcal{S}^2 \times \mathcal{H}^2 \times \mathcal{H}^2_{\nu}\times \mathcal M^{\perp}$: 
 \begin{equation} \label{FBSDEJ_lemma}
    \begin{cases}
    \displaystyle dQ_t^{a^\star}=\left(-a^\star_t-\frac{P_t}{2\a}+\kappa_t\right)dt + \rho_t dW_t + \int_{\mathbb{R}^{\star}} \Xi(t,e)\widetilde N(dt,de), \\
       \displaystyle dY^{a^\star}_t=-\frac{\b}{2}[Q^{a^\star}_t-\overline{Q}_0]dt + Z^{a^\star}_tdW_t + \int_{\R^{\star}}\Gamma^{a^\star}(t,e) \widetilde{N}(dt,de)+H^{a^\star}_t, \\
      \displaystyle  Y_T=\frac{\gamma}{2}Q_T^{a^\star} - \frac{\gamma}{2}\overline Q_0.
    \end{cases}
\end{equation} 
and  the following condition holds : 
\begin{equation} \label{optimality_condition}
     \frac{\a}{2} a^\star_t  -Y_t^{ a^\star} = 0,\ \forall t \in [0,T], \, \text{a.s.}
\end{equation}
Conversely, assume that there exists $(a^\star, Q^{a^\star}, Y^{a^\star}, Z^{a^\star},\Gamma^{a^\star},H^{a^\star}) \in \mathcal{A} \times \mathcal{S}^2 \times \mathcal{S}^2 \times \mathcal{H}^2 \times \mathcal{H}^2_\nu\times \mathcal M^\perp$ satisfying \eqref{FBSDEJ_lemma}-\eqref{optimality_condition}, then $a^\star$ is the optimal control minimizing $\eqref{opt}$.

\end{lemma}
\begin{proof}
Let $a^\star \in \mathcal{A}$ be the optimal control. Define $a^{\varepsilon}:=a^\star + \varepsilon a$ for $a \in \mathcal{A}$ and $\varepsilon>0$. Note that $a^{\varepsilon} \in \mathcal{A}$. Let $Q^\varepsilon \in \mathcal{S}^2$ be the unique strong solution of the following controlled SDE:
\begin{equation*}
    dQ^{\varepsilon}_t=\left(-a^\star_t-\e a_t - \frac{P_t}{2\a}+\kappa_t\right)dt + \rho_t dW_t+ \int_{\mathbb{R}^{\star}} \Xi(t,e)\widetilde N(dt,de), \,\, Q^{\e}_0=\overline{Q}_0.
\end{equation*}
Note that 
\begin{equation*}
    d\left(\frac{Q_t^{\varepsilon}-Q_t^{a^\star}}{\varepsilon}\right)=-a_tdt
\end{equation*}
Applying the Itô's formula to $Y_t^{a^\star}\frac{Q_t^{\varepsilon}-Q_t^{a^\star}}{\varepsilon}$ and taking the expectation, we get:
\begin{multline} \label{ito}
    0=\E\int_0^T Y_t^{a^\star}a_tdt + \E\int_0^T\frac{Q^{\e}_t-Q^{a^\star}_t}{\e}\left(\frac{\b}{2}Q^{a^\star}_t-\frac{\b}{2}\overline Q_0\right)dt  \\+ \E\left[\left(\frac{\g}{2}Q_T^{a^\star} - \frac{\g}{2}\overline Q_0\right)\frac{Q^{\e}_T-Q^{a^\star}_T}{\e}\right].
\end{multline}
 Developing the squares and using \eqref{ito} yields
\begin{align}\label{derivative}
    &\frac{J(0,\overline Q_0,a^{\e})- J(0,\overline Q_0,a^\star)}{\e}=  
    2 \E\left[ \int_0^T \left(\frac{\a}{2} a^\star_t a_t +\frac{\b}{2}Q^{a^\star}_t\frac{Q^{\e}_t-Q^{a^\star}_t}{\e} - \frac{\b}{2}\overline Q_0\frac{Q^{\e}_t-Q^{a^\star}_t}{\e} \right) dt \right. \nonumber \\&\hspace*{4.7cm} \left.+  \frac{\g}{2}Q^{a^\star}_T\frac{Q^{\e}_T-Q^{a^\star}_T}{\e}-\frac{\g}{2}\overline Q_0\frac{Q^{\e}_T-Q^{a^\star}_T}{\e}\right] 
    \nonumber \\& \hspace*{4.7cm}+ \e \E\left[ \int_0^T \frac{\a}{2}a_t^2 + \frac{\b}{2}\left( \frac{Q^{\e}_t-Q^{a^\star}_t}{\e} \right)^2 dt +\frac{\gamma}{2}\left(\frac{Q^{\e}_T-Q^{a^\star}_T}{\e}\right)^2 \right] \nonumber \\
    &=2\E\left[\int_0^Ta_t\left(\frac{\a}{2} a^\star_t  -Y_t^{a^\star}\right)dt\right] +    \e \E\left[\int_0^T\frac{\a}{2}a^2_t + \frac{\b}{2}\left(\frac{Q^{\e}_t-Q^{a^\star}_t}{\e}\right)^2dt + \frac{\g}{2}\left(\frac{Q^{\e}_T-Q^{a^\star}_T}{\e}\right)^2\right]
\end{align}
By taking the limit $\varepsilon \to 0 $ in the above relation and by optimality of $a^\star$, it follows that  \begin{equation*}
    \E\left[\int_0^Ta_t\left(\frac{\a}{2} a^\star_t  -Y_t^{a^\star}\right)dt\right] = 0
\end{equation*}
Since the above equality holds for any control $a \in \mathcal{A}$, we deduce that
\begin{equation}\label{lim2} 
\frac{\a}{2} a^\star_t  -Y_t^{a^\star} = 0,  \,  \forall t \in [0,T] \ \text{a.s.}
\end{equation}
Conversely, assume that there exists $(a^\star, Q^{a^\star}, Y^{a^\star}, Z^{a^\star},\Gamma^{a^\star},H^{a^\star}) \in \mathcal{A} \times \mathcal{S}^2 \times \mathcal{S}^2 \times \mathcal{H}^2 \times \mathcal{H}^2_\nu\times \mathcal M^\perp$ satisfying \eqref{FBSDEJ_lemma}-\eqref{optimality_condition}. From $\eqref{derivative}$, we get that the Gâteaux derivative of $J(0, \overline{Q}_0,\cdot)$ is $0$ at $a^\star$. From the strict convexity of the functional $J(0,\overline{Q}_0,\cdot)$, we deduce  that $a^\star$ is the optimal control. 

\end{proof}

Using the lemma above, we now turn to the proof of the main theorem.
\begin{proof}[Proof of Theorem \ref{main}]
Inspired by \cite{kohlmann_global_2002}, we look for the optimal control $a^\star$ in the following feedback form:
 \begin{align} \label{optimal_control}
 a^\star_t&=\frac{2}{\a}\left(y_tQ_t^{a^\star} + \psi_t\right)  \\  
     dQ_t^{a^\star}&=\left(-a^\star_t-\frac{P_t}{2\a}+\kappa_t\right)dt+ \rho_tdW_t + \int_{\mathbb{R}^{\star}} \Xi(t,e)\widetilde N(dt,de).
 \end{align}
Here, $y_t$ is the unique solution of the following Riccati ODE:
    \begin{align} \label{riccati} 
    &dy_t=-\left(\frac{\b}{2} - \frac{2}{\a}y_t^2\right)dt; \quad y_T=\frac{\g}{2},
    \end{align}
which is given explicitly by
  \begin{equation*}
    y_t=\frac{\sqrt{\a \b}}{2} \frac{1-\frac{\sqrt{\a \b}-\g}{\sqrt{\a \b}+\g}e^{-2\omega\left(T-t\right)}}{1+\frac{\sqrt{\a \b}-\g}{\sqrt{\a \b}+\g}e^{-2\omega\left(T-t\right)}},
 \end{equation*}
and $\psi$ is  the first component of the process $(\psi_t, \phi_t, \xi_t,H_t) \in \mathcal{S}^2 \times \mathcal{H}^2 \times \mathcal{H}^2_\nu\times \mathcal M^\perp$,  which represents  the unique solution of the following linear Backward Stochastic Differential Equation with Jumps (BSDEJ):
\begin{align*}
      d\psi_t &=\left[\frac{2}{\a}y_t\psi_t +\left(\frac{P_t}{2\a}-\kappa_t\right)y_t + \frac{\b}{2}\overline Q_0 \right]dt + \phi_t dW_t+\int_{\mathbb{R}^{\star}} \xi_t(e)\widetilde N(dt,de) + H_t, \\  \psi_T&=-\frac{\g}{2}\overline Q_0.
\end{align*}
Observe that $(\psi_t)$ admits the following representation
 \begin{equation*}
      \psi_t=-\frac{\g}{2}\overline Q_0e^{\frac{2}{\a}A_t}-\E\left[ \int_t^T e^{-\frac{2}{\a}A_{t,s}}\left(\left[\frac{P_s}{2\a}-\kappa_s\right]y_s+\frac{\b}{2}\overline Q_0\right)ds | \mathcal{F}_t\right],
 \end{equation*}
 where 
 \begin{multline*}
     A_t:=-\int_t^T y_s ds =
     \\ \frac{\sqrt{\a\b}}{2}\left(t-T-\sqrt{\frac{\a}{\b}}\left[\log\left(1+\frac{\sqrt{\a \b}-\g}{\sqrt{\a \b}+\g}e^{-2\omega(T-t)}\right)-\log\left(1+\frac{\sqrt{\a \b}-\g}{\sqrt{\a \b}+\g}\right)\right]\right)
 \end{multline*}
  and $A_{t,s}=A_s-A_t$.

In view of Lemma \ref{lemma}, we need to prove that the process $Y_t^{a^\star}:=y_tQ^{a^\star}_t + \psi_t$ satisfies FBSDEJ \eqref{FBSDEJ_lemma} to conclude that $a^\star$ is the optimal control. The terminal condition holds by construction. Applying Itô formula yields
 \begin{align*}
     dY_t^{a^\star}&= y_tdQ_t^{a^\star}+\frac{dy_t}{dt}  Q_t^{a^\star} dt +d\psi_t \\ 
     &=\left[-a^\star_ty_t+\frac{2}{\a}y_t\psi_t+\frac{\b}{2}\overline Q_0 +\frac{dy_t}{dt} Q_t^{a^\star}\right]dt +Z_tdW_t\\
     &+\int_{\mathbb{R}^{\star}} \Gamma_t(e)\widetilde N(dt,de)+H_t, 
 \end{align*}
 where $Z_t:=\phi^1_t+y_t \rho_t$ and $\Gamma_t(e):=y_t\Xi_t(e)+\xi_t(e)$.

 Using equations \eqref{optimal_control} and \eqref{riccati},
 \begin{equation*}
     -a^\star_ty_t+\frac{2}{\a}y_t\psi_t +\frac{dy_t}{dt} Q_t^{a^\star}=-\frac{\b}{2}Q_t^{a^\star},
 \end{equation*}
which proves that $(Y_t^{a^\star},Z_t,\Gamma_t,H_t)_t$ is the unique solution of the FBSDEJ \eqref{FBSDEJ_lemma}. Therefore, by Lemma \ref{lemma}, we conclude that  $a^\star$ is the optimal control.

Now, we go back to the initial control problem \eqref{Optimal}, and by using \eqref{change_var}, we  derive 
\begin{equation}
    q^\star_t=\frac{2}{\a}\left\{y_tQ^{q^\star}_t +\psi_t \right\} + \frac{P_t}{2\a}.
\end{equation}
with 
\begin{equation}
    dQ^{q^\star}_t=(-q_t^\star+\kappa_t)dt + \rho_t dW_t +\int_{\mathbb{R}^{\star}} \Xi(t,e)\widetilde N(dt,de)
\end{equation}
We now proceed to compute the explicit representation of the optimal control process $(q_t^\star)$. To simplify notation, set $u:=\frac{\sqrt{\a \b}-\g}{\sqrt{\a \b}+\g}$. We then obtain, from the representations for $y_t$ ad $\psi_t$ given above,
\begin{multline}
    q^\star_t=\frac{2}{\a}\left\{\frac{\sqrt{\a \b}}{2} \frac{1-ue^{-2\omega(T-t)}}{1+ue^{-2\omega(T-t)}}Q^{q^\star}_t -\frac{\g}{2}\overline Q_0e^{\frac{2}{\a}A_t}- \right. \\ \left.\E\left[ \int_t^T e^{-\frac{2}{\a}A_{t,s}}\left(\left[\frac{P_s}{2\a}-\kappa_s\right]y_s+\frac{\b}{2}\overline Q_0\right)ds | \mathcal{F}_t\right] \right\} + \frac{P_t}{2\a} 
\end{multline}
Since $$-\frac{2}{\a}A_{t,s}=-\omega(s-t)+\log\left(\frac{1+ue^{-2\omega(T-s)}}{1+ue^{-2\omega(T-t)}}\right),$$
we get 

\begin{equation*}
    \int_t^T  e^{-\frac{2}{\a}A_{t,s}} ds=  \sqrt{\frac{\a}{\b}} \frac{1-(1-u)e^{-\omega(T-t)}-ue^{-2\omega(T-t)}}{1+ue^{-2\omega(T-t)}}.
\end{equation*}
Moreover, 

\begin{equation*}
    e^{\frac{2}{\a}A_t}=e^{-\omega(T-t)}\frac{1+u}{1+ue^{-2\omega(T-t)}}
\end{equation*}
Hence, 
\begin{multline}
    q^\star_t=\omega \frac{1-ue^{-2\omega\left(T-t\right)}}{1+ue^{-2\omega\left(T-t\right)}}(Q^{q^\star}_t-\overline Q_0) + \frac{\omega(1-u)-\frac{\g}{\a}(1+u)}{1+ue^{-2\omega(T-t)}}e^{-\omega(T-t)}\overline Q_0 -  \\ \frac{2}{\a} \E\left[ \int_t^T e^{-\frac{2}{\a}A_{t,s}}\left(\frac{P_s}{2\a}-\kappa_s\right)y_sds | \mathcal{F}_t\right] + \frac{P_t}{2\a}. 
\end{multline}

By performing some simple computations, one can remark that $\omega(1-u)-\frac{\g}{\a}(1+u)=0$, leading to
\begin{equation}
     q^\star_t=\omega \frac{1-ue^{-2\omega\left(T-t\right)}}{1+ue^{-2\omega\left(T-t\right)}}\left(Q^{q^\star}_t-\overline Q_0\right) -\frac{2}{\a} \E\left[ \int_t^T e^{-\frac{2}{\a}A_{t,s}}\left(\frac{P_s}{2\a}-\kappa_s\right)y_sds | \mathcal{F}_t\right]  + \frac{P_t}{2\a}
\end{equation}

Set : 
\begin{equation*}
    f(t,T):=\frac{1-ue^{-2\omega\left(T-t\right)}}{1+ue^{-2\omega\left(T-t\right)}}.
\end{equation*}
\begin{equation*}
    f_1(t,s,T):=\frac{2}{\alpha}e^{-\frac{2}{\a}A_{t,s}}y_s=\omega\frac{e^{-\omega(s-t)}-ue^{-\omega(2T-s-t)}}{1+ue^{-2\omega(T-t)}}.
\end{equation*} 

Therefore, $q^\star$ admits the following closed-loop representation :
\begin{equation}
    q^\star_t=\omega f(t,T)\left(Q^{q^\star}_t-\overline Q_0\right) - \E\left[ \int_t^T f_1(t,s,T)\left(\frac{P_s}{2\a}-\kappa_s\right)ds | \mathcal{F}_t\right]  + \frac{P_t}{2\a}.
\end{equation}
\end{proof}

\section{Technical proofs for the deterministic case}
\subsection{Proof of Proposition 3.1}
\label{proof3.1}
In this case, the optimal strategy of a storage agent can be rewritten as : 

\begin{equation}
    q_t=\omega f(t,T)Q^q_t -  \int_t^T f_1(t,s,T)\frac{P_s}{2\a}ds   + \frac{P_t}{2\a}
\end{equation}
This closed-loop form is a first-order ordinary differential equation since $q_t=-\frac{dQ_t}{dt}$. Its solution is given by :
\begin{equation*}
    Q^q_t=e^{-G(t)}\left(\int_0^te^{G(s)}\int_s^Tf_1(s,r,T)\frac{P_r}{2\a}drds-\int_0^te^{G(s)}\frac{P_s}{2\a}ds\right),
\end{equation*}
with $G(t)=\omega t- \log(1+ue^{-2\o (T-t)})$.
Hence, by using Fubini's theorem, 

\begin{multline} \label{determinproof1}
    Q^q_t=e^{-\o t}(1+ue^{-2\o (T-t)}) \left\{\int_0^T\int_0^re^{G(s)}f_1(s,r,T)\frac{P_r}{2\a}dsdr \right.\\ \left. - \int_t^T\int_t^re^{G(s)}f_1(s,r,T)\frac{P_r}{2\a}drds - \int_0^te^{G(s)}\frac{P_s}{2\a}ds\right\}
\end{multline}

We compute the different terms appearing in the above expression separately.
For the first one, we get
\begin{equation}\label{eq1}
    \int_0^re^{G(s)}f_1(s,r,T)ds=\int_0^r \frac{\o e^{\o s}(e^{-\o(r-s)}-ue^{-\o(2T-r-s)})
    }{(1+ue^{-2\o(T-s)})^2}ds
\end{equation}
\begin{align*}
    &=e^{-\o r}\int_0^r \frac{\o e^{2\o s}}{(1+ue^{-2\o(T-s)})^2}ds - e^{\o r} \int_0^r \frac{\o u e^{-2\o (T-s)}}{(1+ue^{-2\o(T-s)})^2}ds
  \\  &=-\frac{1}{2}\frac{\frac{e^{\o (2T-r)}}{u}-e^{\o r}}{1+ue^{-2\o(T-r)}} -\frac{1}{2}\frac{e^{\o r}-\frac{e^{\o (2T-r)}}{u}}{1+ue^{-2\o T}}.
\end{align*}
Similarly, for the second term, we obtain:
\begin{equation}\label{eq2}
    \int_t^re^{G(s)}f_1(s,r,T)ds = -\frac{1}{2}\frac{\frac{e^{\o (2T-r)}}{u}-e^{\o r}}{1+ue^{-2\o(T-r)}} -\frac{1}{2}\frac{e^{\o r}-\frac{e^{\o (2T-r)}}{u}}{1+ue^{-2\o (T-t)}}
\end{equation}
Hence, by replacing \eqref{eq1} and \eqref{eq2}  in (\ref{determinproof1}), we derive:


\begin{align}\label{eq4}
    Q^q_t&=e^{-\o t}(1+ue^{-2\o (T-t)})\left\{-\int_0^t \frac{e^{\o r}}{1+ue^{-2\o (T-r)}}\frac{P_r}{2\a}dr\right. \nonumber \\ &\quad\left. -\int_0^T \frac{1}{2}\frac{\frac{e^{\o (2T-r)}}{u}-e^{\o r}}{1+ue^{-2\o(T-r)}}\frac{P_r}{2\a}dr -\int_0^T \frac{1}{2}\frac{e^{\o r}-\frac{e^{\o (2T-r)}}{u}}{1+ue^{-2\o T}} \frac{P_r}{2\a}dr \right. \nonumber \\ &\quad\left.+ \int_t^T \frac{1}{2}\frac{\frac{e^{\o (2T-r)}}{u}-e^{\o r}}{1+ue^{-2\o(T-r)}}\frac{P_r}{2\a} dr \right\} 
    +\int_t^T \frac{1}{2}\left(e^{\o (r-t)}-\frac{e^{\o (2T-r-t)}}{u}\right)\frac{P_r}{2\a}dr.
\end{align}
Moreover,
\begin{multline}\label{eq3}
    -e^{-\o t}(1+ue^{-2\o (T-t)})\int_0^T \frac{1}{2}\frac{e^{\o r}-\frac{e^{\o (2T-r)}}{u}}{1+ue^{-2\o T}} \frac{P_r}{2\a}dr= \\\frac{1}{2}\int^T_0\left\{-e^{\o(T-t)}\frac{e^{\o(r-T)}-\frac{e^{\o(T-r)}}{u}}{1+ue^{-2\o T}}+ e^{-\o(T-t)} \frac{e^{\o(T-r)}-ue^{\o(r-T)}}{1+ue^{-2\o T}} \right\}\frac{P_r}{2\a} dr
\end{multline}
Recall that 
\begin{align*}
    c_1(P)&= \frac{\gamma\int_0^T \cosh(\omega(T-s)) \frac{P_s}{\alpha} ds+ \beta\int_0^T \frac{\sinh(\omega(T-s))}{\omega} \frac{P_s}{\alpha} ds}{\gamma  \omega\sinh (\omega T)+\beta \cosh (\omega T)} \\
    u&=\frac{\sqrt{\a\b}-\g}{\sqrt{\a\b}+\g}=\frac{\b-\g \o}{\g \o +\b}.
\end{align*}
Hence, the expression \eqref{eq3} is equal to :
\begin{equation}\label{determinproof2}
    -\frac{1}{2}\int^T_0e^{\o(T-t)}\frac{e^{\o(r-T)}-\frac{e^{\o(T-r)}}{u}}{1+ue^{-2\o T}}\frac{P_r}{2\a}dr + {\frac{1}{2}e^{\o t}\o c_1(P)}.
\end{equation}

Moreover, 
\begin{equation*}
    -\frac{1}{2}\int^T_0e^{\o(T-t)}\frac{e^{\o(r-T)}-\frac{e^{\o(T-r)}}{u}}{1+ue^{-2\o T}}\frac{P_r}{2\a}dr + \int_t^T \frac{1}{2}\left(e^{\o {(r-t)}}-\frac{e^{\o (2T{-r-t})}}{u}\right)\frac{P_r}{2\a}dr 
\end{equation*}

\begin{align}
    &=-\frac{1}{2}\int_0^t\left(e^{\o (r-t)}-\frac{e^{\o (2T-r-t)}}{u}\right)\frac{P_r}{2\a}dr + \frac{1}{2} \int_0^T \frac{ue^{\o(-2T+r-t)}-e^{-\o(r+t)}}{1+ue^{-2\o T}}\frac{P_r}{2\a}dr \nonumber
 \\ \label{determinproof3}
    &=-\frac{1}{2}\int_0^t\left(e^{\o (r-t)}-\frac{e^{\o (2T-r-t)}}{u}\right)\frac{P_r}{2\a}dr - \frac{1}{2}e^{-\o t}\o c_1(P)
\end{align}
By combining (\ref{determinproof2}) with (\ref{determinproof3}) and by using \eqref{eq4}, we get that :
\begin{multline*}
    Q^q_t=\o \sinh(\o t)c_1 +\\ \int_0^t \left\{-\frac{e^{-\o t}\left(1+ue^{-2\o (T-t)}\right)}{2}\frac{\frac{e^{\o (2T-r)}}{u}+e^{\o r}}{1+ue^{-2\o(T-r)}}-\frac{1}{2}\left(e^{\o (r-t)}-\frac{e^{\o (2T-r-t)}}{u}\right)\right\}\frac{P_r}{2\a}dr 
\end{multline*}
\begin{equation*}
    =\int_0^t \frac{-e^{\o(t-r)}-e^{\o(r-t)}-ue^{\o(-2T+3r-t)}-ue^{\o(-2T+r+t)}}{2(1+ue^{-2\o(T-r)})} \frac{P_r}{2\a}dr + \o \sinh(\o t)c_1(P).
\end{equation*}

Note that $-(1+ue^{-2\o(T-r)})(e^{\o(r-t)}+e^{\o(t-r)})=-e^{\o(t-r)}-e^{\o(r-t)}-ue^{\o(-2T+3r-t)}-ue^{\o(-2T+r+t)}$.

Hence, we can conclude :

\begin{equation}
    Q^q_t=\o \sinh(\o t)c_1(P) -\int_0^t \cosh(\o(t-r)) \frac{P_r}{2\a}dr.
\end{equation}

\subsection{Proof of Theorem \ref{toy.thm}}
\label{prooftoy.sec}
The equilibrium equation writes:
\begin{align}\label{equation}
D_t - C_0 - C P_t = -c_1(P) \omega^2\cosh\left(\omega t\right) + \frac{P_t}{2\alpha}+\int_0^t \omega\sinh(\omega(t-s)) \frac{P_s}{2\alpha} ds
\end{align}
We first solve it for fixed $c_1(P)$. Introduce the function
$$
G^P(t) = \int_0^t \sinh(\omega(t-s)) P_s ds
$$
Then,
$$
\dot G^P(t) = \int_0^t \omega\cosh(\omega(t-s)) P_s ds\quad \text{and}\quad \ddot G^{P}(t) = \omega P_t + \omega^2 G(t)
$$
Therefore, the equilibrium equation is a 2nd order inhomogenous ordinary differential equation for $G^P$:
$$
 D_t - C_0 + c_1(P) \omega^2\cosh\left(\omega t\right) = \frac{(C+1/2\alpha)}{\omega}(\ddot G^{P}-\omega^2 G^P(t)) + \frac{\omega}{2\alpha} G^P(t),
$$
or in other words
$$
g(t):=\frac{\omega( D_t - C_0 +  c_1(P) \omega^2 \cosh(\omega t))}{C + 1/2\alpha}  =   \ddot G^{P}(t) - \frac{C}{C + 1/2\alpha}  \omega^2 G^P(t)
$$
With initial conditions $G^P(0) = \dot G^P(0) = 0$, the solution takes the form
$$
G^P(t) = \frac{1}{\tilde\omega}\int_0^t \sinh\left(\tilde\omega(t-s)\right) g(s) ds.
$$
Substituting $g(t)$ in the above formula, we get:
\begin{multline*}
  G^P(t) = \frac{\omega}{\tilde\omega(C+1/2\alpha)}\int_0^t\left( D_s - C_0\right) \sinh\left(\tilde\omega(t-s)\right) ds \\ +  c_1(P)\frac{\omega^3}{\tilde\omega(C+1/2\alpha)}\int_0^t \cosh(\omega s) \sinh(\tilde\omega(t-s)) ds.
\end{multline*}
From expression \eqref{represc1}, 
$$
c_1(P):= \frac{1}{2\alpha\omega}\frac{\gamma\dot G^P(T)+ \beta G^P(T)}{\gamma  \omega\sinh\left(\omega T\right)+\beta \cosh\left(\omega T\right)}.
$$
Observe that 
$$
G^P(T) = A + c_1(P)B\quad \text{and}\quad \dot G^P(T) = A' + c_1(P) B'.
$$
Substituting this into the above expression for $c_1(P)$, we obtain a linear equation for this quantity, which yields the formula \eqref{c1}. The function $G^P(t)$ with $P$ equal to the equilibrium price is denoted by $X(t)$, while the value of $c_1(P)$ is denoted by $\Tilde{c}_1$.



Finally, from the equilibrium equation $\eqref{equation}$, we derive
{$$
P_t = \frac{ D_t - C_0 - \frac{\omega}{2\alpha} X(t) + \tilde c_1 \omega^2 \cosh(\omega t)}{C+ \frac{1}{2\alpha}}.
$$}
The withdrawal rate is given by

\begin{align*}
q_t &=  D_t - C_0 - C P_t 
\\&={\frac{( D_t - C_0)}{1+2\alpha C} + \frac{\alpha C}{1+2\alpha C} \left\{\frac{\omega}{\alpha}X(t) - \tilde c_1 \omega^2 \cosh\left(\omega t\right)\right\},}
\end{align*}
and the state of charge is
\begin{align*}
    Q_t^q &= -\int_0^t q_s ds \\
    &= -\int_0^t \frac{( D_s - C_0)ds}{1+2\alpha C} - \frac{\omega C}{1+2\alpha C} \int_0^t X(s)ds+ \frac{\tilde c_1\beta C}{\omega(1+2\alpha C)} \sinh\left(\omega t\right),
\end{align*}
where
\begin{align*}
\int_0^t X(s) ds &= \frac{\omega}{\tilde\omega(C+1/2\alpha)}\int_0^t \int_0^s\left( D_r - C_0\right) \sinh(\tilde\omega(s-r)) drds \\ &+  c_1\frac{\omega^3}{\tilde\omega(C+1/2\alpha)}\int_0^t  \int_0^s \cosh(\omega r) \sinh(\tilde\omega(s-r)) drds\\
 & = \frac{1}{\omega C}\int_0^t\left( D_s - C_0 + \tilde c_1 \omega^2 \cosh\left(\omega s\right)\right)\left\{\cosh\left(\tilde\omega(t-s)\right))-1 \right\}ds.
\end{align*}

\section{Calibration of parameters for the stochastic case} \label{calibration}
Suppose we have a time series of observed hourly data $(\Delta_t) \in \R^{24l}$  over $l$ days. Define $\overline \Delta^h$:
\begin{equation*}
    \overline \Delta_h=\frac{1}{l}\sum_{j \equiv h \text{ mod }24} \Delta_j.
\end{equation*}

Define now $\widetilde \Delta := \Delta- \overline \Delta$. Define $\widetilde \Delta^{OE}$ which follows an Ornstein Uhlenbeck dynamic: 
\begin{equation*}
    \widetilde \Delta^{OE}_t= \widetilde \Delta^{OE}_0 - \int_0^t \theta^{\Delta}(\widetilde \Delta^{OE}_t-\mu^{\Delta})dt + \int_0^t \s^{\Delta}dW^{\Delta}_t 
\end{equation*}
with $W^{\Delta}$ a Brownian motion. We estimate $(\theta^{\Delta},\mu^{\Delta},\sigma^{\Delta})$ through MLE on $\widetilde \Delta$. Moreover, define the following auxiliary process $(\overline \Delta^{OE}_t)$ on $[0,24 \text h ]$:

\begin{equation*}
   \forall h=1,.., 23, \forall t\in [h,h+1[, \quad \overline \Delta^{OE}_t := \overline \Delta_h.
\end{equation*}
Finally, the OE-estimator of $\Delta$ is defined for all $t$ by:
\begin{equation*}
    \Delta^{OE}_t:= \overline \Delta^{OE}_t + \widetilde \Delta^{OE}_t.
\end{equation*}
\section{Scenario coefficients}
\label{RTE_coef}
\begin{table}[H] \label{table_renewable}
\centering
\begin{tabular}{|c|c|c|c|c|c|}
\hline
\textbf{Year} & \textbf{N03} & \textbf{N2} & \textbf{N1} & \textbf{M23} & \textbf{M1} \\ \hline
2019           & 1            & 1           & 1           & 1            & 1           \\ \hline
2030           & 1.92         & 1.92        & 1.92        & 1.92         & 1.92        \\ \hline
2040           & 2.58         & 2.84        & 3.35        & 3.83         & 4.29        \\ \hline
2050           & 3.16         & 3.95        & 4.77        & 5.50         & 6.62        \\ \hline
\end{tabular}
\caption{Multiplication coefficients for renewable generation in the RTE scenarios compared with current generation.}
\end{table}

\begin{table}[H] \label{table_storage}
\centering
\begin{tabular}{|c|c|c|c|c|c|}
\hline
\textbf{Year} & \textbf{N03} & \textbf{N2} & \textbf{N1} & \textbf{M23} & \textbf{M1} \\ \hline
2019           & 5            & 5           & 5           & 5            & 5           \\ \hline
2030           & 5.5          & 5.5         & 5.5         & 5.5          & 5.5         \\ \hline
2040           & 7.1          & 8           & 8.1         & 11.6         & 15.7        \\ \hline
2050           & 9            & 10.5        & 17.2        & 21.2         & 29.1        \\ \hline
\end{tabular}
\caption{Installed storage capacity in the electricity market.}
\end{table}

\begin{table}[H]
\centering

\label{table_demand}
\begin{tabular}{|c|c|c|c|c|}
\hline
\textbf{Year} & \textbf{2019} & \textbf{2030} & \textbf{2040} & \textbf{2050}\\ \hline
Energy Demand  & 1             & 1.07          & 1.19     &  1.36     \\ \hline
\end{tabular}
\caption{Multiplication coefficients for the energy demand in the RTE scenarios compared with current demand.}
\end{table}


 \end{document}